\documentclass[10pt]{article}

\usepackage[english]{babel} 
\usepackage{amsmath}                
\usepackage{amsfonts}               
\usepackage{amssymb}                
\usepackage{amsopn}                 
\allowdisplaybreaks[3] 
\usepackage{amsthm}                
\usepackage{mathrsfs}               
\usepackage[retainorgcmds]{IEEEtrantools} 
\usepackage{calc}                   
\usepackage{graphicx}        
\usepackage{verbatim}               
\usepackage{macros}

\newtheorem{theorem}{Theorem}

\newtheorem{corollary}{Corollary}

\title{On the Capacity of the Dither-Quantized Gaussian Channel}

\author{Tobias Koch}

\date{}

\begin{document}

\maketitle

\begin{abstract}
This paper studies the capacity of the peak-and-average-power-limited Gaussian channel when its output is quantized using a dithered, infinite-level, uniform quantizer of step size $\Delta$.  It is shown that the capacity of this channel tends to that of the unquantized Gaussian channel when $\Delta$ tends to zero, and it tends to zero when $\Delta$ tends to infinity. In the low signal-to-noise ratio (SNR) regime, it is shown that, when the peak-power constraint is absent, the low-SNR asymptotic capacity is equal to that of the unquantized channel irrespective of $\Delta$. Furthermore, an expression for the low-SNR asymptotic capacity for finite peak-to-average-power ratios is given and evaluated in the low- and high-resolution limit. It is demonstrated that, in this case, the low-SNR asymptotic capacity converges to that of the unquantized channel when $\Delta$ tends to zero, and it tends to zero when $\Delta$ tends to infinity. Comparing these results with achievability results for (undithered) 1-bit quantization, it is observed that the dither reduces capacity in the low-precision limit, and it reduces the low-SNR asymptotic capacity unless the peak-to-average-power ratio is unbounded.
\renewcommand{\thefootnote}{}
\footnote{This research was supported by a Marie Curie FP7 Integration Grant within the 7th European Union Framework Programme under Grant 333680 and by the Spanish Government (TEC2009-14504-C02-01, CSD2008-00010, and TEC2012-38800-C03-01).

T.~Koch is with the Signal Theory and Communications Department, Universidad Carlos III de Madrid, Legan\'es 28911, Spain (e-mail: \texttt{koch@tsc.uc3m.es}).}
\end{abstract}
\setcounter{footnote}{0}

\section{Introduction}
\label{sec:intro}
We study the capacity of the discrete-time, peak-and-average-power-limited, Gaussian channel when its output is quantized using a dithered, infinite-level, uniform quantizer of step size $\Delta$ and analyze its behavior in the low- and high-precision limit, where $\Delta$ tends to infinity and zero, respectively.

The problem of quantization arises in communication systems where the receiver uses digital signal processing techniques, so the analog received signal must be sampled and then quantized using an analog-to-digital converter (ADC). If the received signal is sampled at Nyquist rate or above, and if an ADC with high precision is employed, then the effects of sampling and quantization are negligible. However, high-precision ADCs may not be practical when the bandwidth of the system is large and the sampling-rate is high \cite{walden99}. In such scenarios, low-resolution ADCs must be used.

To better understand what communication rates can be achieved with low-resolution ADCs and Nyquist sampling, various works have studied the discrete-time Gaussian channel when its output is quantized using a 1-bit quantizer. At low signal-to-noise ratio (SNR), where communication at low spectral efficiencies takes place, it is known that a symmetric threshold quantizer\footnote{A threshold quantizer produces $1$ if its input is above a threshold, and it produces $0$ if its not. A \emph{symmetric} threshold quantizer is a threshold quantizer whose threshold is zero.} reduces capacity by a factor of $2/\pi$, corresponding to a 2 dB power loss \cite{viterbiomura79}, \cite{singhdabeermadhow09_2}. Hence the rule of thumb that ``hard decisions cause a 2 dB power loss." It was recently demonstrated that this power loss can be avoided by using asymmetric threshold quantizers and asymmetric signal constellations \cite{kochlapidoth13}. However, this result requires \emph{flash-signaling} input distributions \cite[Th.~3]{kochlapidoth13} (see \cite[Def.~2]{verdu02} for a definition). Since such inputs are known to have a poor spectral efficiency \cite[Th.~16]{verdu02}, it follows that for small yet positive spectral efficiencies, the potential  power gain is significantly smaller than 2 dB. For example, at spectral efficiencies of 0.001 bits/s/Hz, allowing for asymmetric quantizers with corresponding asymmetric signal constellations provides a power gain of merely 0.1 dB \cite[Sec.~V]{kochlapidoth13}.

In the following, we refer the Gaussian channel with ($K$-bit) output quantization as the \emph{($K$-bit) quantized Gaussian channel} and to the Gaussian channel without output quantization simply as the \emph{Gaussian channel}. For the Gaussian channel, binary antipodal inputs outperform flash-signaling inputs in terms of spectral efficiency \cite[Th.~11]{verdu02}. However, for such inputs quantizing the channel output with a 1-bit quantizers incurs again a 2 dB power loss at low SNR, since in this case a symmetric threshold quantizer becomes asymptotically optimal as the SNR tends to zero \cite[Prop.~2]{kochlapidoth13}.

Recalling that the discrete-time Gaussian channel arises from the continuous-time, bandlimited, additive white Gaussian noise (AWGN) channel by sampling the output at Nyquist rate, it can be shown that, for binary antipodal signaling and a symmetric threshold quantizer, the 2 dB power loss can be reduced by sampling the channel output above the Nyquist rate. For instance, it was demonstrated that, at low SNR, sampling the output at twice the Nyquist rate improves the power loss from 2 dB for Nyquist sampling to less than 1.28 dB \cite[Th.~1]{kochlapidoth10_2}, \cite[Th.~1]{kochlapidoth10_1_arxiv}. Further results on the capacity of the 1-bit quantized Gaussian channel and super-Nyquist sampling include \cite{zhang12}\nocite{gilbert93}--\cite{shamai94}. Specifically, Zhang \cite{zhang12} studies the generalized mutual information of this channels for a Gaussian codebook ensemble and the nearest-neighbor decoding rule and demonstrates \emph{inter alia} that, as the sampling rate tends to infinity, the power loss is not larger than 0.98 dB. Shamai \cite{shamai94} considers the noiseless case and demonstrates that the capacity is unbounded in the sampling rate. However, it is unknown whether for a symmetric threshold quantizer the power loss can be fully avoided by letting the sampling rate tend to infinity.

Going beyond 1-bit quantizers, it was shown that, at low SNR, a uniform 3-bit quantizer and binary antipodal signaling achieves about 95\% of the capacity of the Gaussian channel, corresponding to a power loss of merely 0.223 dB \cite[Eq.~(3.4.21)]{viterbiomura79}. The capacity of the $K$-bit quantized Gaussian channel was studied, e.g., in \cite{singhdabeermadhow09_2}. The numerical results obtained in \cite{singhdabeermadhow09_2} suggest that, at 0 dB SNR, a 2-bit quantizer achieves still 95\% of the capacity of the Gaussian channel, while at 20 dB SNR, a 3-bit quantizer achieves still 85\% of the capacity of Gaussian channel. However, to the best of our knowledge, there exists no closed-form expression for the capacity of the $K$-bit quantized Gaussian channel, except for the binary case where the channel output is quantized using a symmetric threshold quantizer \cite[Th.~2]{singhdabeermadhow09_2}.

A ubiquitous quantizer is the \emph{uniform quantizer}, whose levels are equispaced, say $\Delta$ apart, either with an infinite or a finite number of levels. We refer to \cite{grayneuhoff98} for a comprehensive survey of quantization theory. For finite-level uniform quantizers, the outermost cells will be semi-infinite and the input space corresponding to these cells is referred to as the \emph{overload region} \cite{grayneuhoff98}. While infinite-level uniform quantizers need an infinite number of bits to describe their output and seem therefore impractical, they have the advantage of eliminating the overload region and resulting overload distortion \cite[Sec.~II-C]{grayneuhoff98}. For this reason, infinite-level uniform quantizers are typically preferred in theoretical analyses, in the hope that the tail of the source to be quantized decays sufficiently fast so the overload distortion be negligible. By Shannon's source coding theorem \cite{shannon48}, irrespective of the number of levels, the output of a uniform quantizer can be described by a variable-length code whose expected length is roughly the entropy of the quantizer output. Consequently, the rate of a quantizer is often measured by the entropy of its output.

The step size $\Delta$ of the uniform quantizer determines its precision:  the smaller $\Delta$, the higher the precision. The high-precision limit (where $\Delta\downarrow 0$) was studied by Gish and Pierce \cite{gishpierce68}, who showed that the difference between the entropy of the output of an infinite-level uniform quantizer and the rate distortion function converges to $\frac{1}{2}\log\frac{\pi e}{6}$ as the permitted distortion (and hence also $\Delta$) vanishes. As for the low-precision limit (where $\Delta\to\infty$), it was shown that for exponential, Laplacian, and Gaussian sources the entropy of the quantizer output approaches zero with the same slope as the rate-distortion function as the allowed distortion tends to the source variance, whereas for uniform sources the slope of the entropy of the quantized output becomes infinite, in contrast to the rate-distortion function which has a finite slope \cite{sullivan96}--\nocite{marconeuhoff06}\cite{gyorgylinder00}. To prove their result for Gaussian sources \cite{marconeuhoff06}, Marco and Neuhoff showed that, in the low-precision limit, the entropy of the quantizer output is determined by the probabilities corresponding to the innermost cells, which is in agreement with the intuition that if the tail of the source decays sufficiently fast, then the overload distortion can be neglected \cite[Lemma 3]{marconeuhoff06}.

A common strategy to further simplify the theoretical analysis of uniform quantizers is \emph{dithering}. (We refer again to \cite[Sec.~V-E]{grayneuhoff98} for a survey of this topic.) In a dithered quantizer, instead of quantizing an input signal directly, one quantizes the sum of the signal and a random process (called a \emph{dither}) that is independent of the signal. This allows one to describe the quantization noise by additive uniform noise that is independent of the input signal. Specifically, if the dither is uniformly distributed over $[-\Delta/2,\Delta/2]$, then the conditional entropy of the quantizer output given the dither is equal to the mutual information between the quantizer input and the sum of the input and independent, uniformly distributed noise \cite[Th.~1]{zamirfeder92}. Dithered quantization was studied in numerous works. Of particular interest to us is the work by Zamir and Feder \cite{zamirfeder95}, which studied the rate-distortion behavior when a bandlimited stationary source is first sampled at Nyquist rate or faster, then it undergoes dithered uniform quantization, and finally it is entropy-encoded. Generalizations of dithered quantization can be found, e.g., in \cite{liklejsakleijn10}, \cite{saldilinderyuksel13}.

Observe that analyses of the capacity of the quantized Gaussian channel are motivated by the need for low-resolution quantizers and therefore typically consider quantizers with a small number of levels. However, the analysis of such quantizers becomes intractable as quantizer resolution and/or sampling rate increase. 
In contrast, theoretical work on quantization often considers infinite-level uniform quantizers, since they allow for a simplified analysis. In this paper, we bring together these two lines of research by studying the capacity of the Gaussian channel when its output is quantized using a dithered, infinite-level, uniform quantizer of step size $\Delta$. (We shall refer to this channel as the \emph{dither-quantized Gaussian channel}.) Since a dithered quantizer can be described as an additive noise channel with uniform noise, the dither-quantized Gaussian channel is equivalent to an additive noise channel where the noise is the sum of a Gaussian and a uniform random variable. This simplifies the analysis of its capacity. While beyond the scope of this paper, we hope that, in the long term, studying the capacity of the dither-quantized Gaussian channel will help us better understand the tradeoff in channel capacity between sampling rate and quantization resolution of the continuous-time, bandlimited, AWGN channel.

The rest of this paper is organized as follows. Section~\ref{sec:channel} introduces the channel model and defines the capacity as well as the low-SNR asymptotic capacity. Section~\ref{sec:capacity} presents the results (as well as the proofs thereof) that concern channel capacity. Section~\ref{sec:lowSNR} presents the results (as well as the proofs thereof) that concern the low-SNR asymptotic capacity. Section~\ref{sec:conclusion} concludes the paper with a summary and a discussion of our results.

\section{Channel Model and Capacity}
\label{sec:channel}
\begin{figure}
\begin{center}
\includegraphics[width=0.8\textwidth]{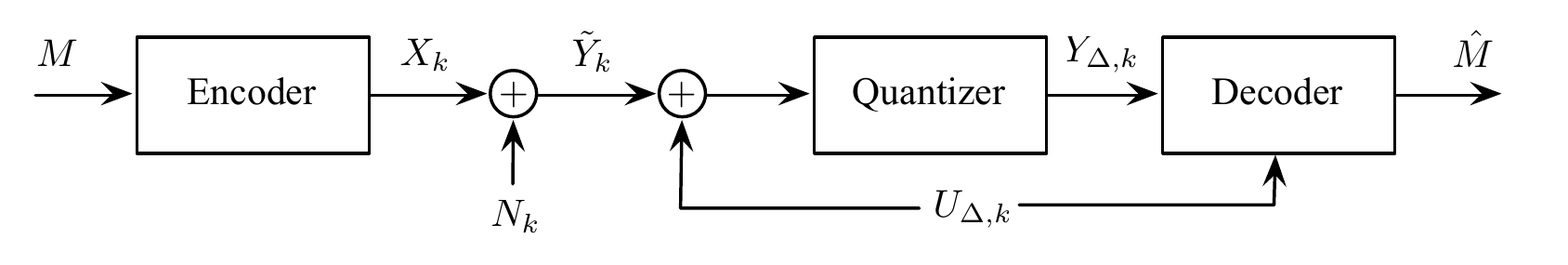}
\end{center}
\label{fig:system}
\caption{System model.}
\end{figure}

We consider the discrete-time communication system depicted in Figure~1. A message $M$, which is uniformly distributed over the set $\{1,\ldots,\mathsf{M}\}$, is mapped by an encoder to the length-$n$ real sequence $X_1,\ldots,X_n\in\Reals$ of channel inputs. (Here, $\Reals$ denotes the set of real numbers.) The channel corrupts this sequence by adding Gaussian noise to produce the unquantized output sequence
\begin{equation}
\label{eq:channel1}
\tilde{Y}_k = X_k + N_k, \quad k\in\Integers
\end{equation}
where $\{N_k,\,k\in\Integers\}$ is a sequence of independent and identically distributed (i.i.d.) Gaussian random variables of mean zero and variance $\sigma^2$. (Here, $\Integers$ denotes the set of integers.) The unquantized sequence is then quantized using a dithered, infinite-level, uniform quantizer of step size $\Delta$. Specifically, the quantizer is a function $q_{\Delta}\colon \Reals \to \Integers$ that produces $i$ if $x\in[i\Delta,(i+1)\Delta)$, i.e.,
\begin{equation}
\label{eq:quantizer}
q_{\Delta}(x) = \left\lfloor\frac{x}{\Delta}\right\rfloor, \quad x\in\Reals
\end{equation}
where, for every $a\in\Reals$, $\lfloor a \rfloor$ denotes the largest integer not larger than $a$.\footnote{In the quantization literature, it is common to consider quantizers whose reproduction values are in the center of their cells, i.e., $q_{\Delta}(x)=\left\lfloor\frac{x}{\Delta}\right\rfloor+\Delta/2$, $x\in\Reals$, since this choice minimizes the expected squared error. For ease of exposition, we use the slightly simpler definition \eqref{eq:quantizer}. In any case, the actual reproduction values do not affect the achievable information rates.} The quantizer output $Y_{\Delta,k}$ is given by
\begin{equation}
\label{eq:channel2}
Y_{\Delta,k} = q_{\Delta}\bigl(\tilde{Y}_k+U_{\Delta,k}\bigr), \quad k\in\Integers
\end{equation}
where $\{U_{\Delta,k},\,k\in\Integers\}$ is a sequence of i.i.d.\ random variables that are uniformly distributed over the interval $[-\Delta/2,\Delta/2]$, referred to as \emph{dither}. We assume that channel input, additive Gaussian noise, and dither are independent. The decoder observes the quantizer output $Y_{\Delta,1},\ldots,Y_{\Delta,n}$ as well as the dither $U_{\Delta,1},\ldots,U_{\Delta,n}$ and guesses which message was transmitted.

We impose both an average-power and a peak-power constraint on the channel inputs: for every realization of $M$, the sequence $x_1,\ldots,x_n$ must satisfy
\begin{equation}
\frac{1}{n} \sum_{k=1}^n x_k^2 \leq \mathsf{P} \quad \text{and} \quad |x_k|^2 \leq \mathsf{A}^2.
\end{equation}
The capacity of the dither-quantized Gaussian channel \eqref{eq:channel1}--\eqref{eq:channel2} under the power constraints $\mathsf{P}$ and $\mathsf{A}^2$ on the channel inputs is given by \cite[Sec.~7.3]{gallager68}
\begin{equation}
\label{eq:capacity}
C_{\Delta}(\mathsf{P},\mathsf{A}) = \sup I(X;Y_{\Delta}|U_{\Delta})
\end{equation}
where the supremum is over all distributions of $X$ satisfying $\E{X^2}\leq\mathsf{P}$ and $|X|\leq\mathsf{A}$ with probability one.\footnote{To account for the dither, we use the standard approach of treating it as an additional channel output that is independent of the channel input.} Here and throughput the paper, we omit the time indices where they are immaterial. When the peak-power constraint is relaxed ($\mathsf{A}=\infty$), we shall denote the capacity by $C_{\Delta}(\mathsf{P})$. In an analogous manner, we shall denote the capacity of the Gaussian channel under the power constraints $\mathsf{P}$ and $\mathsf{A}$ by $C(\mathsf{P},\mathsf{A})$, i.e.,
\begin{equation}
C(\mathsf{P},\mathsf{A}) = \sup I(X;X+N)
\end{equation}
where the supremum is over all distributions of $X$ satisfying $\E{X^2}\leq\mathsf{P}$ and $|X|\leq\mathsf{A}$ with probability one. We shall omit the second argument when the peak-power constraint is relaxed, i.e., $C(\mathsf{P})=C(\mathsf{P},\infty)$. By the data processing inequality \cite[Th.~2.8.1]{coverthomas91},
\begin{equation}
\label{eq:lowres_UB}
C_{\Delta}(\mathsf{P},\mathsf{A}) \leq C(\mathsf{P},\mathsf{A}).
\end{equation}
While it is well-known that the input distribution achieving $C(\mathsf{P},\mathsf{A})$ is discrete \cite{smith71}, to the best of our knowledge, there exists no closed-form expression for $C(\mathsf{P},\mathsf{A})$. Nevertheless, by relaxing the peak-power constraint, we obtain for every $\mathsf{P}$ and $\mathsf{A}$ \cite{shannon48}
\begin{equation}
\label{eq:DPI}
C(\mathsf{P},\mathsf{A}) \leq C(\mathsf{P}) = \frac{1}{2}\log\left(1+\frac{\mathsf{P}}{\sigma^2}\right).
\end{equation}
Here and throughout this paper, $\log(\cdot)$ denotes the natural logarithm function. (Consequently, all rates are in nats per channel use.) In Section~\ref{sub:main_capacity}, we demonstrate that the inequality in \eqref{eq:lowres_UB} becomes tight as $\Delta\downarrow 0$ and that $C_{\Delta}(\mathsf{P},\mathsf{A})$ tends to zero as $\Delta\to\infty$.

Since a dithered quantizer can be described as an additive noise channel with uniform noise $U_{\Delta}$, the dither-quantized Gaussian channel is equivalent to an additive noise channel with noise $Z_{\Delta}=N+U_{\Delta}$. Indeed, following the proof of Theorem~1 in \cite{zamirfeder92}, we show in Appendix~\ref{app:dither_add_noise} that the mutual information on the right-hand side (RHS) of \eqref{eq:capacity} is equal to
\begin{equation}
\label{eq:dither_add_noise}
I(X;Y_{\Delta}|U_{\Delta}) = I(X;X+Z_{\Delta})
\end{equation}
where the probability density function (pdf) of the additive noise $Z_{\Delta}$ is the convolution of the Gaussian and the uniform pdf:
\begin{equation}
\label{eq:pdf_z}
f_{Z_{\Delta}}(z) = \frac{1}{\Delta} \left[Q\left(\frac{z-\Delta/2}{\sigma}\right)-Q\left(\frac{z+\Delta/2}{\sigma}\right)\right].
\end{equation}
Here $Q(\cdot)$ denotes the \emph{Gaussian probability integral} ($Q$-function) \cite[Eq.~(1.3)]{simon02}. 

In addition to capacity, we also study the slope of the capacity-vs-power curve at zero when either the peak-power constraint is relaxed ($\mathsf{A}=\infty$) or when the peak-to-average-power ratio $\mathsf{K}\triangleq\mathsf{A}^2/\mathsf{P}$ is finite and held fixed, i.e.,
\begin{equation}
\label{eq:CUE_infty}
\dot{C}_{\Delta}^{(\infty)}(0) \triangleq \lim_{\mathsf{P}\downarrow 0} \frac{C_{\Delta}(\mathsf{P})}{\mathsf{P}}
\end{equation}
and
\begin{equation}
\label{eq:CUE_K}
\dot{C}_{\Delta}^{(\mathsf{K})}(0) \triangleq \lim_{\mathsf{P}\downarrow 0} \frac{C_{\Delta}(\mathsf{P},\sqrt{\mathsf{K}\mathsf{P}})}{\mathsf{P}}.
\end{equation}
We shall refer to the slope of the capacity-vs-power curve at zero as the \emph{low-SNR asymptotic capacity}.

Relaxing the peak-power constraint allows for a simple expression for $\dot{C}_{\Delta}^{(\infty)}(0)$ \cite[Th.~3]{verdu90}:
\begin{equation}
\label{eq:CUE_verdu}
\dot{C}_{\Delta}^{(\infty)}(0) = \sup_{x\neq 0} \frac{D(P_{X+Z_{\Delta}|X=x}\|P_{X+Z_{\Delta}|X=0})}{x^2}
\end{equation}
where $D(\cdot\|\cdot)$ denotes relative entropy and $P_{X+Z_{\Delta}|X=x}$ denotes the conditional distribution of $X+Z_{\Delta}$ given $X=x$. Unfortunately, $\dot{C}_{\Delta}^{(\infty)}(0)$ may characterize $C_{\Delta}(\mathsf{P})$ only at impractically small input powers $\mathsf{P}$. Indeed, if the supremum on the RHS of \eqref{eq:CUE_verdu} is approached only as $|x|\to\infty$ (as is the case for the 1-bit quantized Gaussian channel \cite[Th.~3]{kochlapidoth13}), then the input distribution that achieves the first derivative of $C_{\Delta}(\mathsf{P})$ at zero (i.e., $\dot{C}_{\Delta}^{(\infty)}(0)$) must be flash signaling, which implies that the second derivative of $C_{\Delta}(\mathsf{P})$ at zero is $-\infty$ \cite{verdu02}. Consequently, in such cases,  $\dot{C}_{\Delta}^{(\infty)}(0)$ describes the behavior of $C_{\Delta}(\mathsf{P})$ poorly, unless $\mathsf{P}$ is very small.

To address this problem, we consider also the case where the peak-to-average-power ratio $\mathsf{K}$ is finite and held fixed, thereby precluding the use of flash signaling input distributions. In this case, it was demonstrated that if the channel law satisfies a number of technical conditions, then the low-SNR asymptotic capacity is given by \cite{ibragimovkhasminskii72}, \cite{prelovvandermeulen93}
\begin{equation}
\label{eq:CUE_prelov}
\dot{C}_{\Delta}^{(\mathsf{K})}(0) = \frac{1}{2} I(0)
\end{equation}
where $I(x)$ denotes the Fisher information
\begin{equation}
\label{eq:FI}
I(x) \triangleq \int_{-\infty}^{\infty} \frac{[\frac{\partial}{\partial x} f_{Z_\Delta}(y-x)]^2}{f_{Z_{\Delta}}(y-x)} \d y.
\end{equation}

By \eqref{eq:lowres_UB} and \eqref{eq:DPI}, and by noting that relaxing the peak-power constraint does not reduce capacity, it follows that
\begin{equation}
\label{eq:ordering}
\dot{C}_{\Delta}^{(\mathsf{K})}(0) \leq \dot{C}_{\Delta}^{(\infty)}(0) \leq \frac{1}{2\sigma^2}.
\end{equation}
In Section~\ref{sub:main_lowSNR}, we demonstrate that the right-most inequality holds with equality irrespective of $\Delta$, while the left-most inequality holds with equality if, and only if, $\Delta$ vanishes.

\section{Channel Capacity}
\label{sec:capacity}
In this section, we study the capacity for arbitrary input powers $\mathsf{P}$ in the high- and low-resolution limit, i.e., when $\Delta\downarrow 0$ and $\Delta\to\infty$, respectively. We show that in the former case, the capacity $C_{\Delta}(\mathsf{P},\mathsf{A})$ converges to that of the Gaussian channel, and in the latter case, it converges to zero.

\subsection{Main Results}
\label{sub:main_capacity}
\begin{theorem}
\label{thm:cap_highres}
Consider the dither-quantized Gaussian channel described in Section~\ref{sec:channel}. Then, for any distribution on $X$ satisfying $\E{X^2}\leq\mathsf{P}$,
\begin{equation}
\label{eq:thm_cap_highres}
\lim_{\Delta\downarrow 0} I(X;X+Z_{\Delta}) = I(X;X+N).
\end{equation}
\end{theorem}
\begin{proof}
Recall that $Z_{\Delta}=N+U_{\Delta}$. To prove Theorem~\ref{thm:cap_highres}, it thus suffices to show that
\begin{IEEEeqnarray}{rCl}
\lim_{\Delta\downarrow 0} h(X+N+U_{\Delta}) & = & h(X+N) \label{eq:thm_highres1}\\
\lim_{\Delta\downarrow 0} h(N+U_{\Delta}) & = & h(N).\label{eq:thm_highres2}
\end{IEEEeqnarray}
Since $N$ is Gaussian and $X$ and $N$ are independent, the differential entropies on the RHS of \eqref{eq:thm_highres1} and \eqref{eq:thm_highres2} are both finite. Furthermore, $\E{N^2}=\sigma^2$, $\E{U_{\Delta}^2}=\frac{\Delta^2}{12}$ and, by the theorem's assumption, $\E{(X+N)^2}\leq \mathsf{P}+\sigma^2$. The above identities \eqref{eq:thm_highres1} and \eqref{eq:thm_highres2} follow therefore directly by specializing the proof of Theorem~1 in \cite{linderzamir94} to the distortion measures $\rho(x)=\delta(x)=x^2$.
\end{proof}

Equation~\eqref{eq:thm_cap_highres} holds for any input distribution satisfying the average-power constraint $\mathsf{P}$, including the capacity-achieving input-distribution. Consequently, Theorem~\ref{thm:cap_highres} implies that the inequality in \eqref{eq:lowres_UB} becomes tight as $\Delta\downarrow 0$.

\begin{corollary}
\label{cor:highres}
Consider the dither-quantized Gaussian channel described in Section~\ref{sec:channel}. Then, for every $\mathsf{P}$ and $\mathsf{A}$,
\begin{equation}
\lim_{\Delta\downarrow 0} C_{\Delta}(\mathsf{P},\mathsf{A}) = C(\mathsf{P},\mathsf{A}).
\end{equation}
\end{corollary}
\begin{proof}
In view of \eqref{eq:lowres_UB}, it suffices to show that
\begin{equation}
\label{eq:cor_highres_LB}
\varliminf_{\Delta\downarrow 0} C_{\Delta}(\mathsf{P},\mathsf{A}) \geq C(\mathsf{P},\mathsf{A})
\end{equation}
where $\varliminf$ denotes the \emph{limit inferior}. To this end, we use that, by Theorem~\ref{thm:cap_highres},  we have for any distribution of $X$ satisfying $\E{X^2}\leq\mathsf{P}$ and $|X|\leq\mathsf{A}$ with probability
\begin{equation}
\varliminf_{\Delta\downarrow 0} C_{\Delta}(\mathsf{P},\mathsf{A}) \geq \lim_{\Delta\downarrow 0} I(X;X+Z_{\Delta}) = I(X;X+N). \label{eq:cor_highres}
\end{equation}
The lower bound \eqref{eq:cor_highres_LB}, and hence Corollary~\ref{cor:highres}, follows by maximizing the RHS of \eqref{eq:cor_highres} over all distributions of $X$ satisfying the power constraints $\mathsf{P}$ and $\mathsf{A}$.
\end{proof}

Theorem~\ref{thm:cap_highres} and Corollary~\ref{cor:highres} demonstrate that, in the high-resolution limit, the dithered quantizer incurs no loss in capacity. As we show next, this is in stark contrast to the low-resolution limit.

\begin{theorem}
\label{thm:cap_lowres}
Consider the dither-quantized Gaussian channel described in Section~\ref{sec:channel}. Then, for every $\mathsf{P}$ and $\mathsf{A}$,
\begin{equation}
\label{eq:thm_cap_lowres}
\lim_{\Delta\to\infty} C_{\Delta}(\mathsf{P},\mathsf{A}) = 0.
\end{equation}
\end{theorem}
\begin{proof}
See Section~\ref{sec:cap_lowres}.
\end{proof}

Let the signal-to-noise-and-quantization-noise-ratio (SNQNR) of the dither-quantized Gaussian channel be defined as
\begin{equation}
\label{eq:SNQNR}
\mathsf{SNQNR} \triangleq \frac{\E{X^2}}{\E{Z^2_{\Delta}}} = \frac{\mathsf{P}}{\sigma^2+\frac{\Delta^2}{12}}.
\end{equation}
Theorem~\ref{thm:cap_lowres} is perhaps not very surprising since the SNQNR tends to zero as $\Delta$ tends to infinity, so it may seem plausible that also the capacity vanishes in the low-resolution limit. However, note that the additive noise $Z_{\Delta}$ is non-Gaussian, so it is \emph{prima facie} unclear whether there is any relation between capacity and SNQNR.

The weak performance of the dithered, infinite-level, uniform quantizer at low quantizer resolutions is due to the dither. Indeed, when the output of the Gaussian channel is quantized using a symmetric threshold quantizer, the capacity is given by \cite[Th.~2]{singhdabeermadhow09_2}
\begin{equation}
\label{eq:capacity_1bit}
C_{\text{1-bit}}(\mathsf{P}) = C_{\text{1-bit}}\bigl(\mathsf{P},\sqrt{\mathsf{P}}\bigr) = \log 2 - H_{\text{b}}\biggl(Q\left(\sqrt{\mathsf{P}}/\sigma\right)\biggr).
\end{equation}
The RHS of \eqref{eq:capacity_1bit} is strictly positive, so this implies that $C_{\text{1-bit}}\bigl(\mathsf{P},\sqrt{\mathsf{P}}\bigr)>\lim_{\Delta\to\infty} C_{\Delta}(\mathsf{P},\mathsf{A})$.
Moreover, since the concatenation of an infinite-level, uniform quantizer and a symmetric threshold quantizer results again in a threshold quantizer, it follows that the \emph{undithered} uniform quantizer achieves a capacity that is at least as large as the capacity achieved by the 1-bit quantizer. Consequently, adding dither is highly detrimental in the low-resolution regime. As we shall see, the same is also true for the low-SNR asymptotic capacity, unless the peak-to-average-power ratio is unbounded.

\subsection{Proof of Theorem~\ref{thm:cap_lowres}}
\label{sec:cap_lowres}
We first note that $\Delta U_1$ has the same distribution as $U_{\Delta}$. Recalling that $Z_{\Delta}=N+U_{\Delta}$, it thus follows that
\begin{equation}
\label{eq:thm_cap_lowres_1}
I(X;X+Z_{\Delta}) = I\left(X;\frac{1}{\Delta}(X+N)+U_1\right).
\end{equation}
We then prove Theorem~\ref{thm:cap_lowres} by showing that
\begin{equation}
\label{eq:thm_cap_lowres_2}
\lim_{\eps\downarrow 0} \sup I\bigl(X;\eps(X+N)+U_1\bigr) = 0
\end{equation}
where the supremum is over all distributions of $X$ satisfying $\E{X^2}\leq\mathsf{P}$ and $|X|\leq\mathsf{A}$ with probability one.

To prove \eqref{eq:thm_cap_lowres_2}, we will follow the steps that were carried out in \cite[Sec.~II]{kochmartinezguillen12} to derive an upper bound on the capacity of the peak-and-average-power-limited complex Gaussian channel. Specifically, we use the upper bound on the mutual information \cite[Th.~5.1]{lapidothmoser03_3}
\begin{equation}
\label{eq:proof_dual}
I(X;Y) \leq \int D\bigl(W(\cdot|x) \bigm\| R(\cdot)\bigr) \d Q(x)
\end{equation}
where $Q(\cdot)$ denotes the input distribution; $W(\cdot|x)$ denotes the conditional distribution of the channel output, conditioned on $X=x$; and $R(\cdot)$ denotes some arbitrary distribution on the output alphabet. Every choice of $R(\cdot)$ yields an upper bound on $I(X;Y)$, and the inequality in \eqref{eq:proof_dual} holds with equality if $R(\cdot)$ is the actual distribution of $Y$ induced by $Q(\cdot)$ and $W(\cdot|\cdot)$. Here, we choose $R(\cdot)$ such that its pdf is
\begin{equation}
\label{eq:r(y)}
r(y) = \begin{cases} \frac{1}{\Upsilon}, & |y|\leq \alpha \\\frac{1}{\Upsilon} \frac{\sqrt{\beta}}{\pi}\frac{1}{1+\beta y^2}, & |y|>\alpha \end{cases}
\end{equation}
for some $\alpha>\frac{1}{2}$ and $0<\beta<1$, where $\Upsilon$ is a normalizing constant, i.e.,
\begin{equation}
\Upsilon \triangleq 2\alpha + 2 \int_{\alpha}^{\infty} \frac{\sqrt{\beta}}{\pi} \frac{1}{1+\beta y^2} \d y = 1+2\left[\alpha-\frac{1}{\pi}\arctan\left(\frac{\alpha}{\sqrt{\beta}}\right)\right] \label{eq:thm_cap_lowres_2.2}
\end{equation}
and $\arctan(\cdot)$ denotes the inverse tangent function. Combining \eqref{eq:r(y)} with \eqref{eq:proof_dual}, and using that conditioning does not increase entropy, we obtain upon substituting $Y=\eps (X+N)+U_1$
\begin{IEEEeqnarray}{lCl}
I\bigl(X;\eps(X+N)+U_1\bigr) & = & - h\bigl(\eps(X+N)+U_1\bigm|X\bigr) - \E{\log r\bigl(\eps(X+N)+U_1\bigr)} \nonumber\\
& \leq & -h\bigl(\epsilon(X+N)+U_1\bigm| X,N\bigr) - \E{\log r\bigl(\eps(X+N)+U_1\bigr)} \nonumber\\
& = & - \E{\log r\bigl(\eps(X+N)+U_1\bigr)} \label{eq:thm_cap_lowres_2.5}
\end{IEEEeqnarray}
where the last step follows because $U_1$ is independent of $(X,N)$, so \cite[Th.~9.6.3]{coverthomas91} and the expression for the differential entropy of a uniform random variable give \[h\bigl(\epsilon(X+N)+U_1\bigm| X,N\bigr) = h(U_1)=0.\]
We next evaluate
\begin{IEEEeqnarray}{lCl}
\IEEEeqnarraymulticol{3}{l}{- \E{\log r\bigl(\eps(X+N)+U_1\bigr)}}\nonumber\\
\quad & = & \log\Upsilon + \Prob(|Y|>\alpha)\log\frac{\pi}{\sqrt{\beta}} + \E{\log\left(1+\beta Y^2\right)\I{|Y|>\alpha}} \label{eq:thm_cap_lowres_3}
\end{IEEEeqnarray}
where $\I{\cdot}$ denotes the indicator function. When $\Prob(|Y|>\alpha)=0$, then \eqref{eq:thm_cap_lowres_3} is equal to
\begin{equation}
\label{eq:thm_cap_lowres_3.5}
- \E{\log r\bigl(\eps(X+N)+U_1\bigr)} = \log\Upsilon
\end{equation}
and \eqref{eq:thm_cap_lowres_2.2}--\eqref{eq:thm_cap_lowres_3.5} give
\begin{equation}
I\bigl(X;\eps(X+N)+U_1\bigr) \leq \log\left(1+2\left[\alpha-\frac{1}{\pi}\arctan\left(\frac{\alpha}{\sqrt{\beta}}\right)\right]\right). \label{eq:thm_cap_lowres_3.7}
\end{equation}
In the following, we consider the case where $\Prob(|Y|>\alpha)>0$. By the triangle inequality, the absolute value of $Y=\eps(X+N)+U_1$ is upper-bounded by $\eps |X+N| + |U_1|$. Furthermore, $|U_1|\leq\frac{1}{2}$. Consequently,
\begin{equation}
\Prob(|Y|>\alpha) \leq \Prob\left(\eps|X+N|>\alpha-\frac{1}{2}\right) \leq \eps^2\frac{\mathsf{P}+\sigma^2}{\left(\alpha-\frac{1}{2}\right)^2} \label{eq:thm_cap_lowres_4}
\end{equation}
where the right-most inequality follows by Chebyshev's inequality \cite[(4.10.7), p.~192]{AsDo00} and because, for every $X$ satisfying $\E{X^2}\leq \mathsf{P}$, we have $\E{|X+N|^2}\leq\mathsf{P}+\sigma^2$. For ease of exposition, we define $\kappa(\alpha)\triangleq \frac{\mathsf{P}+\sigma^2}{(\alpha-\frac{1}{2})^2}$. Since $\log\frac{\pi}{\sqrt{\beta}}>0$, $0<\beta<1$, applying \eqref{eq:thm_cap_lowres_4} to \eqref{eq:thm_cap_lowres_3} thus gives
\begin{IEEEeqnarray}{lCl}
- \E{\log r\bigl(\eps(X+N)+U_1\bigr)} & \leq & \log\Upsilon + \eps^2\kappa(\alpha)\log\frac{\pi}{\sqrt{\beta}} + \E{\log\left(1+\beta Y^2\right)\I{|Y|>\alpha}}. \IEEEeqnarraynumspace\label{eq:thm_cap_lowres_5}
\end{IEEEeqnarray}
To upper-bounded the last term on the RHS of \eqref{eq:thm_cap_lowres_5}, we use Jensen's inequality to obtain
\begin{equation}
\label{eq:thm_cap_lowres_6}
\E{\log\left(1+\beta Y^2\right)\I{|Y|>\alpha}} \leq \Prob(|Y|>\alpha)\log\left(1+\beta \Exp\bigl[Y^2\bigm| |Y|>\alpha\bigr]\right).
\end{equation}
By Bayes' law, we have
\begin{equation}
\label{eq:thm_cap_lowres_7}
\Exp\bigl[Y^2\bigm| |Y|>\alpha\bigr] = \frac{\E{Y^2\I{|Y|>\alpha}}}{\Prob(|Y|>\alpha)} \leq \frac{\eps^2(\mathsf{P}+\sigma^2)+\frac{1}{12}}{\Prob(|Y|>\alpha)}
\end{equation}
where we used in the right-most inequality that $\E{Y^2\I{|Y|>\alpha}}\leq \E{Y^2}$ and that, for every $X$ satisfying $\E{X^2}\leq\mathsf{P}$, the second moment of $Y$ is upper-bounded by $\eps^2(\mathsf{P}+\sigma^2)+\frac{1}{12}$. Combining \eqref{eq:thm_cap_lowres_7} with \eqref{eq:thm_cap_lowres_6} then gives
\begin{IEEEeqnarray}{lCl}
\E{\log\left(1+\beta Y^2\right)\I{|Y|>\alpha}} & \leq & \Prob(|Y|>\alpha)\log\left(1+\beta\frac{\eps^2(\mathsf{P}+\sigma^2)+\frac{1}{12}}{\Prob(|Y|>\alpha)}\right) \nonumber\\
& = & \Prob(|Y|>\alpha)\log\left(\Prob(|Y|>\alpha)+\beta\left[\eps^2(\mathsf{P}+\sigma^2)+\frac{1}{12}\right]\right) \nonumber\\
& & {} -  \Prob(|Y|>\alpha)\log  \Prob(|Y|>\alpha) \nonumber\\
& \leq & \eps^2\kappa(\alpha)\log\left(1+\beta\left[\eps^2(\mathsf{P}+\sigma^2)+\frac{1}{12}\right]\right) + \sup_{0<\xi\leq\eps^2\kappa(\alpha)} |\xi\log\xi| \IEEEeqnarraynumspace \label{eq:thm_cap_lowres_8}
\end{IEEEeqnarray}
where the last step follows by maximizing $-\Prob(|Y|>\alpha)\log  \Prob(|Y|>\alpha)$ over all $\Prob(|Y|>\alpha)$ satisfying \eqref{eq:thm_cap_lowres_4} and because, by \eqref{eq:thm_cap_lowres_4}, $\Prob(|Y|>\alpha)\leq \min\{\eps^2\kappa(\alpha),1\}$.

Combining \eqref{eq:thm_cap_lowres_5} and \eqref{eq:thm_cap_lowres_8} with \eqref{eq:thm_cap_lowres_2.5}, we obtain for $\Prob(|Y|>\alpha)>0$ that
\begin{IEEEeqnarray}{lCl}
\IEEEeqnarraymulticol{3}{l}{I\bigl(X;\eps(X+N)+U_1\bigr)} \nonumber\\
\quad & \leq & \log\Upsilon + \eps^2\kappa(\alpha)\left[\log\frac{\pi}{\sqrt{\beta}}+\log\left(1+\beta\left[\eps^2(\mathsf{P}+\sigma^2)+\frac{1}{12}\right]\right)\right] + \sup_{0<\xi\leq\eps^2\kappa(\alpha)} |\xi\log\xi|. \IEEEeqnarraynumspace\label{eq:thm_cap_lowres_9}
\end{IEEEeqnarray}
Since the RHS of \eqref{eq:thm_cap_lowres_9} is not smaller than the RHS of \eqref{eq:thm_cap_lowres_3.7}, it follows that
\begin{IEEEeqnarray}{lCl}
\IEEEeqnarraymulticol{3}{l}{\sup I\bigl(X;\eps(X+N)+U_1\bigr)} \nonumber\\
\quad & \leq & \log\Upsilon + \eps^2\kappa(\alpha)\left[\log\frac{\pi}{\sqrt{\beta}}+\log\left(1+\beta\left[\eps^2(\mathsf{P}+\sigma^2)+\frac{1}{12}\right]\right)\right] + \sup_{0<\xi\leq\eps^2\kappa(\alpha)} |\xi\log\xi| \IEEEeqnarraynumspace\label{eq:40}
\end{IEEEeqnarray}
where the supremum on the left-hand side (LHS) of \eqref{eq:40} is over all distribution of $X$ satisfying $\E{X^2}\leq\mathsf{P}$ and $|X|\leq\mathsf{A}$ with probability one. Since the function $\xi\mapsto |\xi\log\xi|$ is continuous for $\xi>0$ and vanishes as $\xi\downarrow 0$, it follows that
\begin{equation}
\varlimsup_{\eps\downarrow 0} \sup I\bigl(X;\eps(X+N)+U_1\bigr) \leq \log\left(1+2\left[\alpha-\frac{1}{\pi}\arctan\left(\frac{\alpha}{\sqrt{\beta}}\right)\right]\right) \label{eq:thm_cap_lowres_10}
\end{equation}
where $\varlimsup$ denotes the \emph{limit superior} and where we have substituted $\Upsilon$ by the RHS of \eqref{eq:thm_cap_lowres_2.2}. The claim \eqref{eq:thm_cap_lowres_2}, and hence Theorem~\ref{thm:cap_lowres}, follows from \eqref{eq:thm_cap_lowres_10} by letting first $\beta\downarrow 0$ and then $\alpha\downarrow\frac{1}{2}$.

\section{Low-SNR Asymptotic Capacity}
\label{sec:lowSNR}
In this section, we discuss capacity at low input powers $\mathsf{P}$. We show that, when the peak-power constraint is relaxed, the low-SNR asymptotic capacity is equal to that of the Gaussian channel irrespective of $\Delta$. We further derive an expression for the low-SNR asymptotic capacity for finite peak-to-average-power ratios and evaluate it in the low- and high-resolution limit. We demonstrate that, in this case, the low-SNR asymptotic capacity converges to that of the unquantized channel when $\Delta$ tends to zero, and it tends to zero when $\Delta$ tends to infinity.

\subsection{Main Results}
\label{sub:main_lowSNR}
\begin{theorem}
\label{thm:CUE_thm1}
Consider the dither-quantized Gaussian channel described in Section~\ref{sec:channel}. Then, irrespective of $\Delta>0$,
\begin{equation}
\dot{C}_{\Delta}^{(\infty)}(0) = \frac{1}{2\sigma^2}.
\end{equation}
\end{theorem}
\begin{proof}
See Section~\ref{sec:CUE_thm1}.
\end{proof}

Theorem~\ref{thm:CUE_thm1} is reminiscent of  Theorem~2 in \cite{kochlapidoth13}, which states that the low-SNR asymptotic capacity of the 1-bit quantized Gaussian channel equals $1/(2\sigma^2)$, provided that we allow for flash-signaling input distributions. As noted before, the concatenation of a uniform and a 1-bit quantizer results again in a 1-bit quantizer, so Theorem~\ref{thm:CUE_thm1} may perhaps not be very surprising. However, in general it is unclear how a \emph{dithered} uniform quantizer compares to a 1-bit quantizer, since the dither potentially reduces capacity. In fact, as we shall see next, for finite a peak-to-average-power ratio and as $\Delta$ becomes large, the dither significantly reduces the low-SNR asymptotic capacity.

\begin{theorem}
\label{thm:CUE_thm2}
Consider the dither-quantized Gaussian channel described in Section~\ref{sec:channel}. Then, irrespective of $\mathsf{K}$,
\begin{equation}
\label{eq:CUE_thm2}
\dot{C}_{\Delta}^{(\mathsf{K})}(0) = \frac{1}{\Delta}\frac{1}{4\pi\sigma^2}\int\limits_{-\infty}^{\infty} \frac{\left[e^{-\frac{(y-\Delta/2)^2}{2\sigma^2}}-e^{-\frac{(y+\Delta/2)^2}{2\sigma^2}}\right]^2}{Q\left(\frac{y-\Delta/2}{\sigma}\right)-Q\left(\frac{y+\Delta/2}{\sigma}\right)} \d y.
\end{equation}
\end{theorem}
\begin{proof}
See Section~\ref{sec:CUE_thm2}.
\end{proof}

Observe that for a finite peak-to-average-power ratio, the low-SNR asymptotic capacity depends on $\Delta$. We next study the behavior of $\dot{C}_{\Delta}^{(\mathsf{K})}(0)$ as $\Delta\downarrow 0$ and $\Delta\to\infty$.

\begin{corollary}
\label{cor}
Consider the dither-quantized Gaussian channel described in Section~\ref{sec:channel}. Then,
\begin{subequations}
\begin{IEEEeqnarray}{rrCl}
\text{i)} & \quad \lim_{\Delta\downarrow 0}  \dot{C}_{\Delta}^{(\mathsf{K})}(0) & = & \frac{1}{2\sigma^2} \\
\text{ii)} & \quad \lim_{\Delta\to\infty} \dot{C}_{\Delta}^{(\mathsf{K})}(0) & = & 0.
\end{IEEEeqnarray}
\end{subequations}
\end{corollary}
\begin{proof}
See Section~\ref{sub:proof_cor}.
\end{proof}
Corollary~\ref{cor} demonstrates that, for finite peak-to-average-power ratios, the low-SNR asymptotic capacity of the dither-quantized Gaussian channel approaches that of the Gaussian channel in the high-resolution limit and it vanishes in the low-resolution limit. The latter result is in stark contrast to Proposition~2 in \cite{kochlapidoth13} (see also \cite{viterbiomura79},\cite{singhdabeermadhow09_2}), which demonstrates that for a 1-bit quantizer and $\mathsf{K}=1$, the low-SNR asymptotic capacity equals $1/(\pi\sigma^2)$. Thus, for finite peak-to-average-power ratios, a low-resolution dithered quantizer performs significantly worse than a 1-bit quantizer.

\subsection{Proof of Theorem~\ref{thm:CUE_thm1}}
\label{sec:CUE_thm1}
We shall show that
\begin{equation}
\label{eq:thm1_1}
\sup_{x\neq 0} \frac{D(P_{X+Z_{\Delta}|X=x}\|P_{X+Z_{\Delta}|X=0})}{x^2} \geq \frac{1}{2\sigma^2}.
\end{equation}
Theorem~\ref{thm:CUE_thm1} follows then from \eqref{eq:CUE_verdu}, \eqref{eq:thm1_1}, and \eqref{eq:ordering}. Let
\begin{equation}
V \triangleq \I{X+Z_{\Delta}\geq \Delta\ell_0-\delta}
\end{equation}
for some arbitrary $\ell_0,\delta>0$. By the data processing inequality for relative entropy \cite[Sec.~2.9]{coverthomas91}
\begin{IEEEeqnarray}{lCl}
D(P_{X+Z_{\Delta}|X=x}\|P_{X+Z_{\Delta}|X=0}) & \geq & D(P_{V|X=x}\|P_{V|X=0}) \IEEEeqnarraynumspace \label{eq:thm1_2}
\end{IEEEeqnarray}
where $P_{V|X=x}$ denotes the conditional distribution of $V$ given $X=x$.
Intuitively, $V$ can be viewed as the output of a threshold quantizer with threshold $\Delta\ell_0-\delta$ and input $X+Z_{\Delta}$. Introducing $V$ thus allows us to analyze the RHS of \eqref{eq:thm1_2} following similar steps as the ones reported in \cite[Sec.~VIII-A]{kochlapidoth13}. Indeed, as in \cite[Eq.~(134)]{kochlapidoth13}, we can express the relative entropy as
\begin{IEEEeqnarray}{lCl}
D(P_{V|X=x}\|P_{V|X=0})& = & \left[1-P_{V|X}(1|x)\right] \log\frac{1}{1-P_{V|X}(1|0)} \nonumber\\
& & {} + P_{V|X}(1|x)\log\frac{1}{P_{V|X}(1|0)}  - H_{\text{b}}\bigl(P_{V|X}(1|x)\bigr)\IEEEeqnarraynumspace\,\, \label{eq:thm1_3}
\end{IEEEeqnarray}
where $\log(\cdot)$ denotes the natural logarithm function; $H_{\text{b}}(\cdot)$ denotes the binary entropy function \cite[Eq.~(2.5)]{coverthomas91}; and $P_{V|X}(1|x)\triangleq \Prob(X+Z_{\Delta}\geq \Delta \ell_0-\delta|X=x)$, which can be written as
\begin{equation}
\label{eq:P_V}
P_{V|X}(1|x) = \frac{1}{\Delta} \int_{-\Delta/2}^{\Delta/2} Q\left(\frac{\Delta\ell_0-\delta-x-u}{\sigma}\right) \d u.
\end{equation}
Using that $0<P_{V|X}(1|x)<1$, $x\in\Reals$ and $H_{\text{b}}(p)\leq\log 2$, $0\leq p\leq 1$, \eqref{eq:thm1_3} can be further lower-bounded as
\begin{IEEEeqnarray}{lCl}
D(P_{V|X=x}\|P_{V|X=0}) & \geq & P_{V|X}(1|x)\log\frac{1}{P_{V|X}(1|0)} - \log 2.\label{eq:thm1_4}
\end{IEEEeqnarray}

We next choose $x=\Delta\ell_0 + \Delta/2$ and lower-bound the supremum in \eqref{eq:thm1_1} by letting $\ell_0$ tend to infinity. Together with \eqref{eq:thm1_2} and \eqref{eq:thm1_4}, this yields
\begin{IEEEeqnarray}{lCl}
\sup_{x\neq 0} \frac{D(P_{X+Z_{\Delta}|X=x}\|P_{X+Z_{\Delta}|X=0})}{x^2} & \geq & \lim_{\ell_0\to\infty} \frac{P_{V|X}(1|\Delta\ell_0 + \Delta/2)}{(\Delta\ell_0+\Delta/2)^2} \log\frac{1}{P_{V|X}(1|0)}. \label{eq:thm1_5}
\end{IEEEeqnarray}
By \eqref{eq:P_V} and the monotonicity of the $Q$-function, we obtain
\begin{equation}
\label{eq:thm1_PV1}
P_{V|X}(1|\Delta\ell_0+\Delta/2) \geq Q\left(-\frac{\delta}{\sigma}\right).
\end{equation}
Moreover, by \eqref{eq:P_V} and the following bounds on the $Q$-function \cite[Prop.~19.4.2]{lapidoth09}
\begin{equation}
\label{eq:Q_bounds}
\frac{1}{\sqrt{2\pi x^2}}e^{-x^2/2}\left(1-\frac{1}{x^{2}}\right) < Q(x) < \frac{1}{\sqrt{2\pi x^2}}e^{-x^2/2}, \quad x>0
\end{equation}
we have for sufficiently large $\ell_0$
\begin{equation}
\label{eq:thm1_PV2}
P_{V|X}(1|0) \leq \frac{1}{\sqrt{2\pi}}\frac{\sigma}{\Delta\ell_0-\delta-\Delta/2} e^{-\frac{(\Delta\ell_0-\delta-\Delta/2)^2}{2\sigma^2}}.
\end{equation}
Applying \eqref{eq:thm1_PV1} and \eqref{eq:thm1_PV2} to \eqref{eq:thm1_5} yields
\begin{IEEEeqnarray}{lCl}
\IEEEeqnarraymulticol{3}{l}{\sup_{x\neq 0} \frac{D(P_{X+Z_{\Delta}|X=x}\|P_{X+Z_{\Delta}|X=0})}{x^2}} \nonumber\\
\quad & \geq & \lim_{\ell_0\to\infty}\frac{Q(-\delta/\sigma)}{(\Delta\ell_0+\Delta/2)^2}\Biggl[\log\left(\sqrt{2\pi}\frac{\Delta\ell_0-\delta-\Delta/2}{\sigma}\right) +\frac{(\Delta\ell_0-\delta-\Delta/2)^2}{2\sigma^2}\Biggr] \nonumber\\
& = & Q\left(-\frac{\delta}{\sigma}\right)\frac{1}{2\sigma^2}. \label{eq:thm1_last}
\end{IEEEeqnarray}
The final result \eqref{eq:thm1_1}, and hence Theorem~\ref{thm:CUE_thm1}, follows from \eqref{eq:thm1_last} by letting $\delta$ tend to infinity.

\subsection{Proof of Theorem~\ref{thm:CUE_thm2}}
\label{sec:CUE_thm2}
In order for \eqref{eq:CUE_prelov} to hold, for every $\Delta>0$, the channel law must satisfy six conditions  \cite[Sec.~II]{prelovvandermeulen93}:
\begin{enumerate}
\renewcommand{\labelenumi}{\Alph{enumi}.}
\item The channel can be described by a pdf $f_{Y|X}$.
\item The pdf $f_{Y|X}(y|x)$ is bounded for all $|x|<\epsilon$ (for some $\epsilon>0$) and $y\in\Reals$.
\item The partial derivative $\frac{\partial}{\partial x} f_{Y|X}(y|x)$ exists for all $|x|<\epsilon$ and $y\in\Reals$.
\item The Fisher information \eqref{eq:FI} exists and is finite for all $|x|<\epsilon$.
\item The function $\frac{\partial}{\partial x} \bigl(\sqrt{f_{Y|X}(y|x)}\bigr)$ is uniformly continuous in the mean square with respect to $|x|<\epsilon$.
\item For any $\delta>0$,
\begin{equation}
\label{eq:Condition_F}
\lim_{\epsilon\to 0} \frac{1}{\epsilon} \int_{-\epsilon}^{\epsilon} \int_{\set{B}_{\epsilon,\delta}} \frac{[\frac{\partial}{\partial x} f_{Y|X}(y|x)]^2}{f_{Y|X}(y|x)} \d y \d x = 0
\end{equation}
where
\begin{equation}
\label{eq:Beps}
\set{B}_{\epsilon,\delta}\triangleq \left\{y\in\Reals\colon \sup_{|x|<\epsilon}\left|\log\frac{f_{Y|X}(y|x)}{f_{Y|X}(y|0)}\right|>\delta\right\}.
\end{equation}
\end{enumerate}

Note that, for the channel model described in Section~\ref{sec:channel}, we have $f_{Y|X}(y|x)=f_{Z_{\Delta}}(y-x)$. Thus, Conditions~A and B follow directly by inspecting \eqref{eq:pdf_z}. Furthermore, using that $\frac{\partial}{\partial x} Q(x) = -\frac{1}{\sqrt{2\pi}}e^{-\frac{x^2}{2}}$, we obtain from \eqref{eq:pdf_z} that
\begin{IEEEeqnarray}{lCl}
\frac{\partial}{\partial x} f_{Y|X}(y|x) & = & \frac{1}{\Delta}\frac{1}{\sqrt{2\pi\sigma^2}} \left[e^{-\frac{(y-x-\Delta/2)^2}{2\sigma^2}}-e^{-\frac{(y-x+\Delta/2)^2}{2\sigma^2}}\right] \label{eq:thm2_1}
\end{IEEEeqnarray}
which proves Condition~C. This also demonstrates that the Fisher information \eqref{eq:FI} exists and is given by
\begin{equation}
\label{eq:thm2_2}
I(x) = \frac{1}{\Delta} \frac{1}{2\pi\sigma^2} \int\limits_{-\infty}^{\infty} \frac{\left[e^{-\frac{(y-x-\Delta/2)^2}{2\sigma^2}}-e^{-\frac{(y-x+\Delta/2)^2}{2\sigma^2}}\right]^2}{Q\left(\frac{y-x-\Delta/2}{\sigma}\right)-Q\left(\frac{y-x+\Delta/2}{\sigma}\right)} \d y.
\end{equation}
To prove Condition~D, it thus remains to show that the Fisher information is finite for all $|x|<\epsilon$. This, as well as Conditions~E and F, require slightly more involved proofs, which are presented in Appendix~\ref{app:thm2}. Having proven Conditions~A-F, Theorem~\ref{thm:CUE_thm2} follows directly by combining \eqref{eq:thm2_2} with \eqref{eq:CUE_prelov}.

\subsection{Proof of Corollary~\ref{cor}}
\label{sub:proof_cor}

\subsubsection{Part~i)}
To prove Part~i), we note that, by \eqref{eq:ordering}, we have $\dot{C}_{\Delta}^{(\mathsf{K})}(0)\leq 1/(2\sigma^2)$. It thus remains to show that
\begin{equation}
\label{eq:cor_HR_ASH}
\varliminf_{\Delta\downarrow 0} \dot{C}_{\Delta}^{(\mathsf{K})}(0) \geq \frac{1}{2\sigma^2}.
\end{equation}
To this end, we use Fatou's lemma \cite[(1.6.8), p.~50]{AsDo00} to lower-bound \eqref{eq:CUE_thm2} as
\begin{IEEEeqnarray}{lCl}
\varliminf_{\Delta\downarrow 0} \dot{C}_{\Delta}^{(\mathsf{K})}(0) & \geq & \frac{1}{4\pi\sigma^2}\int\limits_{-\infty}^{\infty} \varliminf_{\Delta\downarrow 0} \frac{1}{\Delta}\frac{\left[e^{-\frac{(y-\Delta/2)^2}{2\sigma^2}}-e^{-\frac{(y+\Delta/2)^2}{2\sigma^2}}\right]^2}{Q\left(\frac{y-\Delta/2}{\sigma}\right)-Q\left(\frac{y+\Delta/2}{\sigma}\right)} \d y.
\end{IEEEeqnarray}
We next apply l'H\^opital's rule twice to compute the limit inside the integral. Indeed, we have
\begin{IEEEeqnarray}{lCl}
\IEEEeqnarraymulticol{3}{l}{\frac{\partial}{\partial\Delta} \left[e^{-\frac{(y-\Delta/2)^2}{2\sigma^2}}-e^{-\frac{(y+\Delta/2)^2}{2\sigma^2}}\right]^2} \nonumber\\
\quad & = & 2\left[e^{-\frac{(y-\Delta/2)^2}{2\sigma^2}}-e^{-\frac{(y+\Delta/2)^2}{2\sigma^2}}\right]\left[\frac{y-\Delta/2}{2\sigma^2}e^{-\frac{(y-\Delta/2)^2}{2\sigma^2}}+\frac{y+\Delta/2}{2\sigma^2}e^{-\frac{(y+\Delta/2)^2}{2\sigma^2}}\right]
\end{IEEEeqnarray}
and
\begin{IEEEeqnarray}{lCl}
\IEEEeqnarraymulticol{3}{l}{\frac{\partial}{\partial\Delta} \Delta\left[Q\left(\frac{y-\Delta/2}{\sigma}\right)-Q\left(\frac{y+\Delta/2}{\sigma}\right)\right]}\nonumber\\
\quad &  = & Q\left(\frac{y-\Delta/2}{\sigma}\right)-Q\left(\frac{y+\Delta/2}{\sigma}\right) + \frac{\Delta}{2\sqrt{2\pi\sigma^2}}\left[e^{-\frac{(y-\Delta/2)^2}{2\sigma^2}}+e^{-\frac{(y+\Delta/2)^2}{2\sigma^2}}\right]
\end{IEEEeqnarray}
which both tend to zero as $\Delta\downarrow 0$. We further have
\begin{IEEEeqnarray}{lCl}
\frac{\partial^2}{\partial\Delta^2}\left[e^{-\frac{(y-\Delta/2)^2}{2\sigma^2}}-e^{-\frac{(y+\Delta/2)^2}{2\sigma^2}}\right]^2 & = & 2\left[\frac{y-\Delta/2}{2\sigma^2}e^{-\frac{(y-\Delta/2)^2}{2\sigma^2}}+\frac{y+\Delta/2}{2\sigma^2}e^{-\frac{(y+\Delta/2)^2}{2\sigma^2}}\right]^2 \nonumber\\
\IEEEeqnarraymulticol{3}{l}{ {} + 2\left[e^{-\frac{(y-\Delta/2)^2}{2\sigma^2}}-e^{-\frac{(y+\Delta/2)^2}{2\sigma^2}}\right] \left[\frac{(y-\Delta/2)^2-\Delta\sigma^2}{4\sigma^4}e^{-\frac{(y-\Delta/2)^2}{2\sigma^2}}-\frac{(y+\Delta/2)^2-\Delta\sigma^2}{4\sigma^4}e^{-\frac{(y+\Delta/2)^2}{2\sigma^2}}\right]} \nonumber\\\label{eq:cor_HR_1}
\end{IEEEeqnarray}
and
\begin{IEEEeqnarray}{lCl}
\frac{\partial^2}{\partial\Delta^2} \Delta\left[Q\left(\frac{y-\Delta/2}{\sigma}\right)-Q\left(\frac{y+\Delta/2}{\sigma}\right)\right] & = & \frac{1}{\sqrt{2\pi\sigma^2}} \left[e^{-\frac{(y-\Delta/2)^2}{2\sigma^2}}+e^{-\frac{(y+\Delta/2)^2}{2\sigma^2}}\right] \nonumber\\
\IEEEeqnarraymulticol{3}{r}{{} + \frac{\Delta}{2\sqrt{2\pi\sigma^2}}\left[\frac{y-\Delta/2}{2\sigma^2}e^{-\frac{(y-\Delta/2)^2}{2\sigma^2}}-\frac{y+\Delta/2}{2\sigma^2}e^{-\frac{(y+\Delta/2)^2}{2\sigma^2}}\right]}.\label{eq:cor_HR_2}
\end{IEEEeqnarray}
Noting that, as $\Delta\downarrow 0$, \eqref{eq:cor_HR_1} and \eqref{eq:cor_HR_2} tend to $2 y^2 e^{-\frac{y^2}{\sigma^2}}/\sigma^4$ and $\sqrt{2/(\pi\sigma^2)}e^{-\frac{y^2}{2\sigma^2}}$, respectively, l'H\^opital's rule gives
\begin{equation}
\lim_{\Delta\downarrow 0} \frac{1}{\Delta}\frac{\left[e^{-\frac{(y-\Delta/2)^2}{2\sigma^2}}-e^{-\frac{(y+\Delta/2)^2}{2\sigma^2}}\right]^2}{Q\left(\frac{y-\Delta/2}{\sigma}\right)-Q\left(\frac{y+\Delta/2}{\sigma}\right)} = \frac{2\frac{y^2}{\sigma^4} e^{-\frac{y^2}{\sigma^2}}}{\sqrt{\frac{2}{\pi\sigma^2}}e^{-\frac{y^2}{2\sigma^2}}} = \sqrt{2\pi\sigma^2}\frac{y^2}{\sigma^4}e^{-\frac{y^2}{2\sigma^2}}. \label{eq:cor_HR_3}
\end{equation}
Integrating \eqref{eq:cor_HR_3} from $-\infty$ to $\infty$ yields
\begin{equation}
\frac{1}{4\pi\sigma^2}\int\limits_{-\infty}^{\infty} \lim_{\Delta\downarrow 0} \frac{1}{\Delta}\frac{\left[e^{-\frac{(y-\Delta/2)^2}{2\sigma^2}}-e^{-\frac{(y+\Delta/2)^2}{2\sigma^2}}\right]^2}{Q\left(\frac{y-\Delta/2}{\sigma}\right)-Q\left(\frac{y+\Delta/2}{\sigma}\right)} \d y = \frac{1}{2\sigma^2}\frac{1}{\sqrt{2\pi\sigma^2}}\int_{-\infty}^{\infty} \frac{y^2}{\sigma^2} e^{-\frac{y^2}{2\sigma^2}} \d y = \frac{1}{2\sigma^2}
\end{equation}
where the last step follows by identifying the terms after $1/(2\sigma^2)$ as the variance of a zero-mean, variance-$\sigma^2$, Gaussian random variable divided by $\sigma^2$, which is equal to one. This proves \eqref{eq:cor_HR_ASH}, which in turn proves Part~i) of Corollary~\ref{cor}.

\subsubsection{Part~ii)}
To prove Part~ii), it suffices to show that the integral on the RHS of \eqref{eq:CUE_thm2} is bounded in $\Delta$. To this end, we first note that the integrand in \eqref{eq:CUE_thm2} is symmetric in $y$, so
\begin{equation}
\frac{1}{4\pi\sigma^2}\int\limits_{-\infty}^{\infty} \frac{\left[e^{-\frac{(y-\Delta/2)^2}{2\sigma^2}}-e^{-\frac{(y+\Delta/2)^2}{2\sigma^2}}\right]^2}{Q\left(\frac{y-\Delta/2}{\sigma}\right)-Q\left(\frac{y+\Delta/2}{\sigma}\right)} \d y = \frac{1}{4\pi\sigma^2}\int\limits_{-\infty}^{\infty} \frac{\left[e^{-\frac{(|y|-\Delta/2)^2}{2\sigma^2}}-e^{-\frac{(|y|+\Delta/2)^2}{2\sigma^2}}\right]^2}{Q\left(\frac{|y|-\Delta/2}{\sigma}\right)-Q\left(\frac{|y|+\Delta/2}{\sigma}\right)} \d y.\label{eq:cor_LR_0}
\end{equation}
We next divide the integration region into the two regions
\begin{equation}
\label{eq:cor_LR_1}
\set{Y}_1=\left\{y\in\Reals\colon |y|\leq\vartheta+\frac{\Delta}{2}\right\} \quad \text{and} \quad \set{Y}_2=\left\{y\in\Reals\colon |y|>\vartheta+\frac{\Delta}{2}\right\}
\end{equation}
for a sufficiently large $\vartheta>0$ and analyze the corresponding integrals separately. 

For $y\in\set{Y}_1$, we use the monotonicity of the $Q$-function to lower-bound
\begin{equation}
Q\left(\frac{|y|-\Delta/2}{\sigma}\right)-Q\left(\frac{|y|+\Delta/2}{\sigma}\right)\geq Q\bigl(\vartheta/\sigma\bigr)-Q\bigl(\Delta/(2\sigma)\bigr), \quad y\in\set{Y}_1 \label{eq:cor_LR_lambda}
\end{equation}
where, for $\vartheta<\Delta/2$, the RHS of \eqref{eq:cor_LR_lambda} is strictly positive and tends to $Q(\vartheta/\sigma)$ as $\Delta\to\infty$. Together with the identity $(a+b)^2\leq 2(a^2+b^2)$, this yields
\begin{IEEEeqnarray}{lCl}
\IEEEeqnarraymulticol{3}{l}{\frac{1}{4\pi\sigma^2}\int_{\set{Y}_1} \frac{\left[e^{-\frac{(|y|-\Delta/2)^2}{2\sigma^2}}-e^{-\frac{(|y|+\Delta/2)^2}{2\sigma^2}}\right]^2}{Q\left(\frac{|y|-\Delta/2}{\sigma}\right)-Q\left(\frac{|y|+\Delta/2}{\sigma}\right)} \d y} \nonumber\\
\quad & \leq & \frac{1}{Q\bigl(\vartheta/\sigma\bigr)-Q\bigl(\Delta/(2\sigma)\bigr)}\frac{1}{2\pi\sigma^2}\int_{\set{Y}_1} e^{-\frac{(|y|-\Delta/2)^2}{\sigma^2}} + e^{-\frac{(|y|+\Delta/2)^2}{\sigma^2}} \d y \nonumber\\
& \leq & \frac{1}{Q\bigl(\vartheta/\sigma\bigr)-Q\bigl(\Delta/(2\sigma)\bigr)}\frac{1}{\sqrt{\pi\sigma^2}} \label{eq:cor_LR_Y1}
\end{IEEEeqnarray}
where the last inequality follows by enhancing the integration region from $\set{Y}_1$ to $\Reals$.

We next consider the case where $y\in\set{Y}_2$. By \eqref{eq:Q_bounds}, we have for $|y|>\vartheta+\Delta/2$\begin{IEEEeqnarray}{lCl}
\IEEEeqnarraymulticol{3}{l}{Q\left(\frac{|y|-\Delta/2}{\sigma}\right)-Q\left(\frac{|y|+\Delta/2}{\sigma}\right)} \nonumber\\
\quad & \geq & \frac{1}{\sqrt{2\pi}}\frac{\sigma}{|y|-\Delta/2}\left(1-\frac{\sigma^2}{(|y|-\Delta/2)^2}\right)e^{-\frac{(|y|-\Delta/2)^2}{2\sigma^2}} - \frac{1}{\sqrt{2\pi}}\frac{\sigma}{|y|+\Delta/2}e^{-\frac{(|y|+\Delta/2)^2}{2\sigma^2}} \nonumber\\
& \geq & \frac{1}{\sqrt{2\pi}}\frac{\sigma}{|y|-\Delta/2} e^{-\frac{(|y|-\Delta/2)^2}{2\sigma^2}}\left[1-\frac{\sigma^2}{\vartheta^2}-\frac{|y|-\Delta/2}{|y|+\Delta/2}e^{-|y|\frac{\Delta}{\sigma^2}}\right] \nonumber\\
& \geq & \frac{\mu_{\vartheta}(\Delta)}{\sqrt{2\pi}}\frac{\sigma}{|y|-\Delta/2} e^{-\frac{(|y|-\Delta/2)^2}{2\sigma^2}}, \quad y\in\set{Y}_2 \label{eq:cor_LR_2}
\end{IEEEeqnarray}
where
\begin{equation}
\label{eq:cor_LR_3}
\mu_{\vartheta}(\Delta) \triangleq 1-\frac{\sigma^2}{\vartheta^2}-e^{-(\vartheta+\Delta/2)\frac{\Delta}{\sigma^2}}
\end{equation}
which, for sufficiently large $\vartheta$, is strictly positive and tends to $1-\sigma^2/\vartheta^2$ as $\Delta\to\infty$. The last inequality in \eqref{eq:cor_LR_2} follows because $(x-\Delta/2)/(x+\Delta/2)\leq 1$, $x>0$ and because the function $e^{-|y|\frac{\Delta}{\sigma^2}}$ is monotonically decreasing in $|y|$. We further note that, for $|y|>\vartheta+\Delta/2$,
\begin{equation}
0 \leq e^{-\frac{(|y|-\Delta/2)^2}{2\sigma^2}}-e^{-\frac{(|y|+\Delta/2)^2}{2\sigma^2}} \leq e^{-\frac{(|y|-\Delta/2)^2}{2\sigma^2}}. \label{eq:cor_LR_4}
\end{equation}
By \eqref{eq:cor_LR_2} and \eqref{eq:cor_LR_4},
\begin{IEEEeqnarray}{lCl}
\IEEEeqnarraymulticol{3}{l}{\frac{1}{4\pi\sigma^2}\int_{\set{Y}_2} \frac{\left[e^{-\frac{(|y|-\Delta/2)^2}{2\sigma^2}}-e^{-\frac{(|y|+\Delta/2)^2}{2\sigma^2}}\right]^2}{Q\left(\frac{|y|-\Delta/2}{\sigma}\right)-Q\left(\frac{|y|+\Delta/2}{\sigma}\right)} \d y} \nonumber\\
\quad & \leq & \frac{1}{2\sqrt{2\pi\sigma^2}} \frac{1}{\mu_{\vartheta}(\Delta)} \int_{\set{Y}_2} \frac{|y|-\Delta/2}{\sigma^2}e^{-\frac{(|y|-\Delta/2)^2}{2\sigma^2}}\d y \nonumber\\
& = & \frac{1}{\sqrt{2\pi\sigma^2}} \frac{1}{\mu_{\vartheta}(\Delta)}e^{-\frac{\vartheta^2}{2\sigma^2}}. \label{eq:cor_LR_Y2}
\end{IEEEeqnarray}
Combining \eqref{eq:cor_LR_Y1} and \eqref{eq:cor_LR_Y2} with \eqref{eq:cor_LR_0}, we obtain
\begin{equation}
\frac{1}{4\pi\sigma^2}\int\limits_{-\infty}^{\infty} \frac{\left[e^{-\frac{(y-\Delta/2)^2}{2\sigma^2}}-e^{-\frac{(y+\Delta/2)^2}{2\sigma^2}}\right]^2}{Q\left(\frac{y-\Delta/2}{\sigma}\right)-Q\left(\frac{y+\Delta/2}{\sigma}\right)} \d y \leq \frac{1}{\sqrt{2\pi\sigma^2}}\left(\frac{\sqrt{2}}{Q\bigl(\vartheta/\sigma\bigr)-Q\bigl(\Delta/(2\sigma)\bigr)} + \frac{e^{-\frac{\vartheta^2}{2\sigma^2}}}{\mu_{\vartheta}(\Delta)}\right). \label{eq:cor_LR_final}
\end{equation}
Part~ii) of Corollary~\ref{cor} follows then by noting that, for sufficiently large $\vartheta$, the RHS of \eqref{eq:cor_LR_final} is bounded in $\Delta$.

\section{Conclusion}
\label{sec:conclusion}
We have studied both the capacity and the low-SNR asymptotic capacity of the peak-and-average-power-limited Gaussian channel when its output is quantized using a dithered, infinite-level, uniform quantizer of step size $\Delta$. We have demonstrated that the capacity of the dither-quantized channel converges to the capacity of the unquantized channel in the high-resolution limit ($\Delta\downarrow 0$), and it converges to zero in the low-resolution limit ($\Delta\to\infty$). We have further demonstrated that, when the peak-power constraint is absent, the low-SNR asymptotic capacity of the dither-quantized channel is equal to that of the unquantized channel irrespective of $\Delta$. In contrast, for finite peak-to-average-power ratios, the low-SNR asymptotic capacity of the dither-quantized channel depends critically on $\Delta$: as we show, it converges to the low-SNR asymptotic capacity of the unquantized channel in the high-resolution limit, but it vanishes in the low-resolution limit.

While dithered, infinite-level, uniform quantizers seem impractical due to the infinite number of bits required to describe their outputs, studying their behavior may help us better understand the behavior of quantizers with a small number of levels, provided that both type of quantizers have similar behaviors. Our results suggest that, with respect to channel capacity, this is the case in the high-resolution limit, but it is not the case in the low-resolution limit. For example, the capacity of the 1-bit quantized Gaussian channel (with a symmetric threshold quantizer) is given by \eqref{eq:capacity_1bit}, which differs not only from the capacity of the dither-quantized Gaussian channel in the low-resolution limit for a given $\mathsf{P}$ (which is zero), but it also has a distinct asymptotic behavior as $\mathsf{P}$ tends to zero. 

Since the concatenation of an infinite-level, uniform quantizer and a 1-bit quantizer results again in a 1-bit quantizer, we conclude that the inferior performance at low quantizer resolutions of the dithered, infinite-level, uniform quantizer is due to the dither. In other words, in the low-resolution regime, adding dither is highly detrimental. Nevertheless, sampling the output of the dithered, infinite-level, uniform quantizer above Nyquist rate, as studied in \cite{zamirfeder95}, may perhaps improve the performance in this regime, since such an approach reduces the quantization noise without increasing the quantizer resolution.

\appendix

\section{Quantization Noise}
\label{app:dither_add_noise}
We shall prove \eqref{eq:dither_add_noise} by showing that
\begin{IEEEeqnarray}{lCl}
H(Y_{\Delta}|U_{\Delta}) & = & h(X+Z_{\Delta}) - \log\Delta \label{eq:app_h|u}\\
H\bigl(Y_{\Delta} |U_{\Delta},X) & = & h(Z_{\Delta}) - \log\Delta. \label{eq:app_h|u,x}
\end{IEEEeqnarray}
The proof of \eqref{eq:app_h|u} and \eqref{eq:app_h|u,x} is almost identical to the proof of Theorem~1 in \cite{zamirfeder92}. For the sake of completeness, we repeat it here.

First note that, since $X$ and $N$ are independent and $N$ is Gaussian, it follows by  \cite[Th.~4.10]{durrett05} that the distribution of the random variable $V=X+N$ is absolutely continuous with respect to the Lebesgue measure, so its pdf, which we shall denote by $f_V$, is defined. Furthermore, the pdf of $V+U_{\Delta}=X+Z_{\Delta}$ relates to $f_V$ via \cite[Th.~4.10]{durrett05}
\begin{equation}
f_{X+Z_{\Delta}}(\xi) = \int_{-\infty}^{\infty} f_{V}(u) f_{U_{\Delta}}(\xi-u) \d u = \frac{1}{\Delta} \int_{\xi-\Delta/2}^{\xi+\Delta/2} f_V(u) \d u \label{eq:app1_1}
\end{equation}
where $f_{U_{\Delta}}$ denotes the pdf of $U_{\Delta}$, i.e., $f_{U_{\Delta}}(u) = \frac{1}{\Delta}\I{|u|\leq\Delta/2}$, $u\in\Reals$. (Recall that $U_{\Delta}$ is uniformly distributed over $[-\Delta/2,\Delta/2]$ and $Z_{\Delta}=N+U_{\Delta}$.) Likewise, the conditional probability of $Y_{\Delta}$ given $U_{\Delta}=u$ is equal to
\begin{equation}
P_{Y_{\Delta}|U_{\Delta}}(i |u) = \Prob\bigl(\Delta i\leq V+u < \Delta (i+1)\bigr) = \int_{\Delta i-u}^{\Delta(i+1)-u} f_V(v) \d v
\end{equation}
which together with \eqref{eq:app1_1} yields
\begin{equation}
\label{eq:app1_2}
P_{Y_{\Delta}|U_{\Delta}}(i |u) = \Delta f_{X+Z_{\Delta}}(\Delta i+\Delta/2-u).
\end{equation}
We next use \eqref{eq:app1_2} and Fubini's theorem \cite[(2.6.6), p.~108]{AsDo00} to express the conditions entropy of $Y$ given $U_{\Delta}$ as
\begin{IEEEeqnarray}{lCl}
H(Y_{\Delta}|U_{\Delta}) & = & - \frac{1}{\Delta} \int_{-\Delta/2}^{\Delta/2} \sum_{i=-\infty}^{\infty} P_{Y|U_{\Delta}}(i |u) \log P_{Y|U_{\Delta}}(i |u) \d u \nonumber\\
& = & - \sum_{i=-\infty}^{\infty} \frac{1}{\Delta} \int_{-\Delta/2}^{\Delta/2} P_{Y|U_{\Delta}}(i |u) \log P_{Y|U_{\Delta}}(i |u) \d u \nonumber\\
& = & - \log\Delta - \sum_{i=-\infty}^{\infty}\int_{-\Delta/2}^{\Delta/2} f_{X+Z_{\Delta}}(\Delta i+\Delta/2-u) \log f_{X+Z_{\Delta}}(\Delta i+\Delta/2-u)\d u. \IEEEeqnarraynumspace\label{eq:app1_h|u}
\end{IEEEeqnarray}
By the change of variable $\xi=\Delta i+\Delta/2-u$, it then follows that
\begin{IEEEeqnarray}{lCl}
H(Y_{\Delta}|U_{\Delta}) & = & - \log\Delta - \sum_{i=-\infty}^{\infty} \int_{\Delta i}^{\Delta(i+1)} f_{X+Z_{\Delta}}(\xi) \log f_{X+Z_{\Delta}}(\xi) \d\xi \nonumber\\
& = & -\log\Delta - h(X+Z_{\Delta}). \label{eq:app1_h|U}
\end{IEEEeqnarray}
This proves \eqref{eq:app_h|u}.

The second identity \eqref{eq:app_h|u,x} follows along similar lines. Indeed, we have
\begin{equation}
f_{Z_{\Delta}}(\xi) = \int_{-\infty}^{\infty} f_N(u) f_{U_{\Delta}}(\xi-u)\d u = \frac{1}{\Delta} \int_{\xi-\Delta/2}^{\xi+\Delta/2} f_N(u) \d u \label{eq:app1_3}
\end{equation}
where $f_N$ denotes the pdf of $N$. Furthermore, the conditional probability of $Y_{\Delta}$ given $(U_{\Delta},X)=(u,x)$ is
\begin{equation}
P_{Y_{\Delta}|U_{\Delta},X}(i|u,x) = \Prob\bigl(\Delta i\leq x+N+u< \Delta(i+1)\bigr) = \int_{\Delta i-x-u}^{\Delta(i+1)-x-u} f_N(n) \d n \label{eq:app1_4}
\end{equation}
which together with \eqref{eq:app1_4} yields
\begin{equation}
\label{eq:app1_5}
P_{Y_{\Delta}|U_{\Delta},X}(i|u,x) = \Delta f_{Z_{\Delta}}(\Delta i +\Delta/2 -x-u).
\end{equation}
Analog to \eqref{eq:app1_h|u} and \eqref{eq:app1_h|U}, we obtain from \eqref{eq:app1_5} that for every $x\in\Reals$
\begin{IEEEeqnarray}{lCl}
H(Y_{\Delta}|U_{\Delta},X=x) & = & -\log\Delta - \sum_{i=-\infty}^{\infty} \int_{\Delta i-x}^{\Delta(i+1)-x} f_{Z_{\Delta}}(\xi)\log f_{Z_{\Delta}}(\xi) \d\xi \nonumber\\
& = & - \log\Delta -h(Z_{\Delta}).
\end{IEEEeqnarray}
Averaging over $X$, this yields \eqref{eq:app_h|u,x}.

\section{Appendix to Section~\ref{sec:CUE_thm2}}
\label{app:thm2}
In this appendix, we prove the conditions stated in Section~\ref{sec:CUE_thm2} that require more involved proofs. Specifically, Section~\ref{sub:Cond_D} demonstrates that the Fisher information \eqref{eq:FI} is finite for all $|x|<\eps$, which together with \eqref{eq:thm2_2} proves Condition~D. Section~\ref{sub:Cond_E} proves Condition~E and Section~\ref{sub:Cond_F} proves Condition~F.

Throughout this appendix, we shall use the following notation. We denote the partial derivative of $f_{Y|X}(y|x)$ with respect to $x$ by $f_x'(y|x)\triangleq \frac{\partial}{\partial x} f_{Y|X}(y|x)$. We further omit the subscript of $f_{Y|X}(y|x)$ to keep notation compact. Finally, we define the sets
\begin{equation}
\set{Y}_1\triangleq \{y\in\Reals\colon |y|\leq \vartheta\} \quad \text{and} \quad \set{Y}_2\triangleq \{y\in\Reals\colon |y|>\vartheta\}.
\end{equation}
 for some arbitrary $\vartheta$.
 
\subsection{Condition~D}
\label{sub:Cond_D}
The Fisher information $I(\cdot)$ is given by \eqref{eq:thm2_2}, namely,
\begin{equation}
I(x) = \frac{1}{\Delta} \frac{1}{2\pi\sigma^2} \int\limits_{-\infty}^{\infty} \frac{\left[e^{-\frac{(y-x-\Delta/2)^2}{2\sigma^2}}-e^{-\frac{(y-x+\Delta/2)^2}{2\sigma^2}}\right]^2}{Q\left(\frac{y-x-\Delta/2}{\sigma}\right)-Q\left(\frac{y-x+\Delta/2}{\sigma}\right)} \d y. \label{eq:appD_0}
\end{equation}
To prove that $I(x)$ is finite for all $|x|<\eps$ and $\Delta>0$, we divide the integration region into $\set{Y}_1$ and $\set{Y}_2$, for some sufficiently large $\vartheta>0$, and show that the corresponding integrals are finite for all $|x|<\eps$.

Since the $Q$-function is continuous and $\set{Y}_1$ is a closed and bounded interval, it follows from the extreme value theorem that for every $y\in\set{Y}_1$ and $|x|\leq \eps$
\begin{equation}
Q\left(\frac{y-x-\Delta/2}{\sigma}\right) - Q\left(\frac{y-x+\Delta/2}{\sigma}\right) \geq Q\left(\frac{\xi_0-\Delta/2}{\sigma}\right) - Q\left(\frac{\xi_0+\Delta/2}{\sigma}\right)
\end{equation}
for some $\xi_0\in[-\vartheta-\eps,\vartheta+\eps]$. Together with \eqref{eq:pdf_z}, this yields for every $y\in\set{Y}_1$ and $|x|<\eps$
\begin{IEEEeqnarray}{lCl}
f(y|x) \geq \frac{1}{\Delta}\left[Q\left(\frac{\xi_0-\Delta/2}{\sigma}\right) - Q\left(\frac{\xi_0+\Delta/2}{\sigma}\right)\right] \triangleq \lambda_{\Delta}. \label{eq:appD_1}
\end{IEEEeqnarray}
By the strict monotonicity of $Q(\cdot)$, it further follows that $\lambda_{\Delta}>0$. We thus have
\begin{IEEEeqnarray}{lCl}
\IEEEeqnarraymulticol{3}{l}{\frac{1}{\Delta} \frac{1}{2\pi\sigma^2} \int_{\set{Y}_1} \frac{\left[e^{-\frac{(y-x-\Delta/2)^2}{2\sigma^2}}-e^{-\frac{(y-x+\Delta/2)^2}{2\sigma^2}}\right]^2}{Q\left(\frac{y-x-\Delta/2}{\sigma}\right)-Q\left(\frac{y-x+\Delta/2}{\sigma}\right)} \d y} \nonumber\\
\quad & \leq & \frac{1}{\Delta\lambda_{\Delta}}\frac{1}{2\pi\sigma^2} \int_{\set{Y}_1} \left[e^{-\frac{(y-x-\Delta/2)^2}{2\sigma^2}}-e^{-\frac{(y-x+\Delta/2)^2}{2\sigma^2}}\right]^2 \d y \nonumber\\
& \leq & \frac{1}{\Delta\lambda_{\Delta}}\frac{\vartheta}{2\pi\sigma^2} \label{eq:appD_Y1}
\end{IEEEeqnarray}
where the second inequality follows because
\begin{equation}
\left|\exp\biggl(-\frac{(y-x-\Delta/2)^2}{2\sigma^2}\biggr)-\exp\biggl(-\frac{(y-x+\Delta/2)^2}{2\sigma^2}\biggr)\right|\leq 1.\label{eq:appD_bla}
\end{equation}

We next consider the case where $y\in\set{Y}_2$. To this end, we first note that the pdf $f_{Z_{\Delta}}$ is symmetric in $z$, so it can be written as
\begin{equation}
f_{Z_{\Delta}}(z) = \frac{1}{\Delta} \left[Q\left(\frac{|z|-\Delta/2}{\sigma}\right)-Q\left(\frac{|z|+\Delta/2}{\sigma}\right)\right].
\end{equation}
Using \eqref{eq:Q_bounds}, this can be lower-bounded as
\begin{IEEEeqnarray}{lCl}
f_{Z_{\Delta}}(z) & \geq & \frac{1}{\Delta}\frac{1}{\sqrt{2\pi}}\frac{\sigma}{|z|-\Delta/2}\left(1-\frac{\sigma^2}{(|z|-\Delta/2)^2}\right) e^{-\frac{(|z|-\Delta/2)^2}{2\sigma^2}} - \frac{1}{\Delta}\frac{1}{\sqrt{2\pi}}\frac{\sigma}{|z|+\Delta/2}e^{-\frac{(|z|+\Delta/2)^2}{2\sigma^2}} \nonumber\\
& = & \frac{1}{\Delta}\frac{1}{\sqrt{2\pi}}\frac{\sigma}{|z|-\Delta/2}e^{-\frac{(|z|-\Delta/2)^2}{2\sigma^2}}\biggl[1-\frac{\sigma^2}{(|z|-\Delta/2)^2} - \frac{|z|-\Delta/2}{|z|+\Delta/2}e^{-\frac{|z|\Delta}{\sigma^2}}\biggr]. \label{eq:appD_2}
\end{IEEEeqnarray}
Note that the term inside the square brackets on the RHS of \eqref{eq:appD_2} tends to one as $|z|\to\infty$. Since by the triangle inequality $|y-x|\geq \vartheta-\epsilon$ for $y\in\set{Y}_2$ and $|x|<\eps$, it follows that for any $0<\mu_{\Delta}<1$ there exists a sufficiently large $\vartheta$ such that
\begin{equation}
f(y|x) = f_{Z_{\Delta}}(y-x) \geq \frac{\mu_{\Delta}}{\Delta}\frac{1}{\sqrt{2\pi}}\frac{\sigma}{|y-x|-\Delta/2}e^{-\frac{(|y-x|-\Delta/2)^2}{2\sigma^2}}, \quad y\in\set{Y}_2, \, |x|<\eps. \label{eq:appD_3}
\end{equation}
Applying \eqref{eq:appD_3} to \eqref{eq:appD_0}, and using that the integrand is symmetric in $y$, we obtain
\begin{IEEEeqnarray}{lCl}
\IEEEeqnarraymulticol{3}{L}{\frac{1}{\Delta} \frac{1}{2\pi\sigma^2}\int_{\set{Y}_2} \frac{\left[e^{-\frac{(y-x-\Delta/2)^2}{2\sigma^2}}-e^{-\frac{(y-x+\Delta/2)^2}{2\sigma^2}}\right]^2}{Q\left(\frac{y-x-\Delta/2}{\sigma}\right)-Q\left(\frac{y-x+\Delta/2}{\sigma}\right)} \d y}\nonumber\\
\quad & \leq &\frac{1}{2\pi\sigma^2}\frac{\sqrt{2\pi}}{\mu_{\Delta}\sigma} \int_{\set{Y}_2}  (|y-x|-\Delta/2)e^{\frac{(|y-x|-\Delta/2)^2}{2\sigma^2}}\left[e^{-\frac{(|y-x|-\Delta/2)^2}{2\sigma^2}}-e^{-\frac{(|y-x|+\Delta/2)^2}{2\sigma^2}}\right]^2 \d y. \nonumber\\
& \leq & \frac{1}{2\pi\sigma^2}\frac{\sqrt{2\pi}}{\mu_{\Delta}\sigma} \int_{\set{Y}_2} (|y-x|-\Delta/2)e^{-\frac{(|y-x|-\Delta/2)^2}{2\sigma^2}} \d y \label{eq:appD_4}
\end{IEEEeqnarray}
where the last step follows because, for sufficiently large $\vartheta$, we have $|y-x|\geq\vartheta - \eps>\Delta/2$, which implies that
\begin{equation}
\left|e^{-\frac{(|y-x|-\Delta/2)^2}{2\sigma^2}}-e^{-\frac{(|y-x|+\Delta/2)^2}{2\sigma^2}}\right| \leq e^{-\frac{(|y-x|-\Delta/2)^2}{2\sigma^2}}. \label{eq:appD_4.5}
\end{equation}
Let $z=y-x$ and $\set{Z}_2=\{z\in\Reals\colon |z+x|>\vartheta\}$. By a change of variables, \eqref{eq:appD_4} can be further upper-bounded by
\begin{IEEEeqnarray}{lCl}
\IEEEeqnarraymulticol{3}{l}{\frac{1}{2\pi\sigma^2}\frac{\sqrt{2\pi}}{\mu_{\Delta}\sigma} \int_{\set{Y}_2} (|y-x|-\Delta/2)e^{-\frac{(|y-x|-\Delta/2)^2}{2\sigma^2}} \d y} \nonumber\\
\quad & = & \frac{1}{2\pi\sigma^2}\frac{\sqrt{2\pi}}{\mu_{\Delta}\sigma} \int_{\set{Z}_2} (|z|-\Delta/2)e^{-\frac{(|z|-\Delta/2)^2}{2\sigma^2}} \d z \nonumber\\
& \leq & \frac{1}{2\pi\sigma^2}\frac{\sqrt{2\pi}}{\mu_{\Delta}\sigma} \int_{|z|>\vartheta-\eps}(|z|-\Delta/2)e^{-\frac{(|z|-\Delta/2)^2}{2\sigma^2}} \d z \nonumber\\
& = & \sqrt{\frac{2}{\pi\sigma^2}} \frac{1}{\mu_{\Delta}} e^{-\frac{(\vartheta-\eps-\Delta/2)^2}{2\sigma^2}} \label{eq:appD_Y2}
\end{IEEEeqnarray}
where the inequality follows because, by the triangle inequality, $\set{Z}_2\subseteq \{z\in\Reals\colon |z|>\vartheta-\eps\}$ and because for $z\in\set{Z}_2$ and a sufficiently large $\vartheta$, the term $(|z|-\Delta/2)$ is nonnegative.

Combining \eqref{eq:appD_Y1} and \eqref{eq:appD_Y2}, we obtain for every $|x|<\eps$ and a sufficiently large $\vartheta$
\begin{equation}
I(x) \leq \frac{1}{\Delta\lambda_{\Delta}}\frac{2\vartheta}{\pi\sigma^2} +  \sqrt{\frac{2}{\pi\sigma^2}} \frac{1}{\mu_{\Delta}} e^{-\frac{(\vartheta-\eps-\Delta/2)^2}{2\sigma^2}}.
\end{equation}
Thus, the Fisher information $I(x)$ is finite for all $|x|<\eps$.

\subsection{Condition~E}
\label{sub:Cond_E}
By the chain rule, it follows that
\begin{equation}
\frac{\partial}{\partial x} \sqrt{f(y|x)} = \frac{f_x'(y|x)}{2\sqrt{f(y|x)}}.
\end{equation}
To prove Condition~E, we need to show that for every $\Delta>0$ \cite[Eq.~(2.3)]{prelovvandermeulen93}
\begin{equation}
\label{eq:app:condE_1}
\int_{-\infty}^{\infty} \left[\frac{f_x'(y|x_1)}{2\sqrt{f(y|x_1)}}-\frac{f_x'(y|x_2)}{2\sqrt{f(y|x_2)}}\right]^2 \d y \to 0
\end{equation}
as $x_1\to 0$ and $x_2\to 0$. Since, for all $y\in\Reals$, $x\mapsto f'_x(y|x)$ and $x\mapsto f(y|x)$ are both bounded and continuous functions of $x$ and $\inf_{|x|<\epsilon} f(y|x)>0$ (for some arbitrary $\epsilon>0$), it follows that
\begin{equation}
\label{eq:appE_2}
\lim_{\substack{x_1\to 0, \\x_2\to 0}} \left[\frac{f_x'(y|x_1)}{2\sqrt{f(y|x_1)}}-\frac{f_x'(y|x_2)}{2\sqrt{f(y|x_2)}}\right]^2 = 0, \quad y\in\Reals.
\end{equation}
To prove \eqref{eq:app:condE_1}, it thus suffices to show that there exists an integrable function $y\mapsto g(y)$ that upper-bounds
\begin{equation}
\label{eq:appE_3}
\left[\frac{f_x'(y|x_1)}{2\sqrt{f(y|x_1)}}-\frac{f_x'(y|x_2)}{2\sqrt{f(y|x_2)}}\right]^2 \leq g(y), \quad y\in\Reals
\end{equation}
for all $|x_1|<\eps$ and $|x_2|<\eps$.The claim \eqref{eq:app:condE_1}, and hence Condition~E, follows then by the dominated convergence theorem \cite[(1.6.9), p.~50]{AsDo00}.

To prove \eqref{eq:appE_3}, we follow the approach carried out in Section~\ref{sub:Cond_D} and divide the integration region into $\set{Y}_1$ and $\set{Y}_2$ (for some sufficiently large $\vartheta>0$) and evaluate the corresponding integrals separately. For $y\in\set{Y}_1$, we use the identity $(a+b)^2\leq 2(a^2+b^2)$, \eqref{eq:appD_1}, and \eqref{eq:appD_bla} to upper-bound
\begin{IEEEeqnarray}{lCl}
\IEEEeqnarraymulticol{3}{l}{\left[\frac{f_x'(y|x_1)}{2\sqrt{f(y|x_1)}}-\frac{f_x'(y|x_2)}{2\sqrt{f(y|x_2)}}\right]^2 }\nonumber\\
\quad & \leq & \frac{\bigl(f_x'(y|x_1)\bigr)^2}{2 f(y|x_1)}+\frac{\bigl(f_x'(y|x_2)\bigr)^2}{2 f(y|x_2)} \nonumber\\
& \leq & \frac{1}{2\lambda_{\Delta}}\frac{1}{\Delta}\frac{1}{2\pi\sigma^2} \left(\left[e^{-\frac{(y-x_1-\Delta/2)^2}{2\sigma^2}}-e^{-\frac{(y-x_1+\Delta/2)^2}{2\sigma^2}}\right]^2+\left[e^{-\frac{(y-x_2-\Delta/2)^2}{2\sigma^2}}-e^{-\frac{(y-x_2+\Delta/2)^2}{2\sigma^2}}\right]^2\right) \nonumber\\
& \leq & \frac{1}{\lambda_{\Delta}}\frac{1}{\Delta}\frac{1}{2\pi\sigma^2} \label{eq:appE_Y1}
\end{IEEEeqnarray}
which is integrable over the bounded set $\set{Y}_1$.

We next consider the case where $y\in\set{Y}_2$. We first note that for any $0<\mu_{\Delta}<1$ there exists a sufficiently large $\vartheta$ such that \eqref{eq:appD_3} holds. Using this result together with the identity $(a+b)^2\leq 2(a^2+b^2)$ and \eqref{eq:appD_4.5}, we obtain for sufficiently larger $\vartheta$
\begin{IEEEeqnarray}{lCl}
\IEEEeqnarraymulticol{3}{l}{\left[\frac{f_x'(y|x_1)}{2\sqrt{f(y|x_1)}}-\frac{f_x'(y|x_2)}{2\sqrt{f(y|x_2)}}\right]^2 }\nonumber\\
\quad & \leq & \frac{\bigl(f_x'(y|x_1)\bigr)^2}{2 f(y|x_1)}+\frac{\bigl(f_x'(y|x_2)\bigr)^2}{2 f(y|x_2)} \nonumber\\
& \leq & \frac{1}{4\pi\sigma^2}\frac{\sqrt{2\pi}}{\mu_{\Delta}\sigma} (|y-x_1|-\Delta/2)e^{\frac{(|y-x_1|-\Delta/2)^2}{2\sigma^2}}\left[e^{-\frac{(|y-x_1|-\Delta/2)^2}{2\sigma^2}}-e^{-\frac{(|y-x_1|+\Delta/2)^2}{2\sigma^2}}\right]^2 \nonumber\\
& & {} + \frac{1}{4\pi\sigma^2}\frac{\sqrt{2\pi}}{\mu_{\Delta}\sigma} (|y-x_2|-\Delta/2)e^{\frac{(|y-x_2|-\Delta/2)^2}{2\sigma^2}}\left[e^{-\frac{(|y-x_2|-\Delta/2)^2}{2\sigma^2}}-e^{-\frac{(|y-x_2|+\Delta/2)^2}{2\sigma^2}}\right]^2 \nonumber\\
& \leq &  \frac{1}{4\pi\sigma^2}\frac{\sqrt{2\pi}}{\mu_{\Delta}\sigma}\left( (|y-x_1|-\Delta/2)e^{-\frac{(|y-x_1|-\Delta/2)^2}{2\sigma^2}} +  (|y-x_2|-\Delta/2)e^{-\frac{(|y-x_2|-\Delta/2)^2}{2\sigma^2}}\right). \IEEEeqnarraynumspace\label{eq:appE_107}
\end{IEEEeqnarray}
Since, by the triangle inequality, $\bigl||y|-|x|\bigr|\leq |y-x|\leq |y|+|x|$, it follows that for all $|x_1|<\epsilon$ and $|x_2|<\epsilon$
\begin{equation}
\left[\frac{f_x'(y|x_1)}{2\sqrt{f(y|x_1)}}-\frac{f_x'(y|x_2)}{2\sqrt{f(y|x_2)}}\right]^2 \leq \frac{1}{2\pi\sigma^2}\frac{\sqrt{2\pi}}{\mu_{\Delta}\sigma} (|y|+\eps-\Delta/2)e^{-\frac{(|y|-\eps-\Delta/2)^2}{2\sigma^2}}. \label{eq:appE_Y2}
\end{equation}
Note that the RHS of \eqref{eq:appE_Y2} is integrable over $y\in\set{Y}_2$.

Combining \eqref{eq:appE_Y1} and \eqref{eq:appE_Y2}, it follows that the integrable function
\begin{equation}
g(y) = \begin{cases}  \frac{1}{\lambda_{\Delta}}\frac{1}{\Delta}\frac{1}{2\pi\sigma^2}, & y\in\set{Y}_1 \\  \frac{1}{2\pi\sigma^2}\frac{\sqrt{2\pi}}{\mu_{\Delta}\sigma} (|y|+\eps-\Delta/2)e^{-\frac{(|y|-\eps-\Delta/2)^2}{2\sigma^2}}, & y\in\set{Y}_2\end{cases}
\end{equation}
satisfies \eqref{eq:appE_3} for all $|x_1|<\eps$ and $|x_2|<\eps$. This demonstrates that Condition~E is satisfied.

\subsection{Condition~F}
\label{sub:Cond_F}
We upper-bound the left-hand side of \eqref{eq:Condition_F} by deriving an upper bound on
\begin{equation*}
\int_{\set{B}_{\epsilon,\delta}} \frac{[f'_x(y|x))]^2}{f(y|x)} \d y
\end{equation*}
that holds for sufficiently small $\eps>0$ and that is independent of $x\in\Reals$ and $\delta>0$. To this end, we divide the integration region into $\set{B}_{\epsilon,\delta}\cap \set{Y}_1$ and $\set{B}_{\epsilon,\delta}\cap \set{Y}_2$ and evaluate each integral separately. We then show that the resulting upper bound vanishes as $\vartheta$ tends to infinity, thereby proving Condition~F.

We begin by showing that for any $\delta>0$ and $\vartheta>0$ there exists a sufficiently small $\eps_0$ such that for all $\eps<\eps_0$
\begin{equation}
\label{eq:appF_0}
\set{B}_{\epsilon,\delta}\cap\set{Y}_1 = \varnothing
\end{equation}
where $\varnothing$ denotes the empty set. Consequently,
\begin{equation}
\label{eq:appF_Y1}
\int_{\set{B}_{\epsilon,\delta}\cap\set{Y}_1} \frac{[f'_x(y|x)]^2}{f(y|x)} \d y=0, \quad \eps<\eps_0.
\end{equation}
To this end, we approximate $f(y|x)$ for every $y\in\set{Y}_1$ by a Taylor series around $x=0$:
\begin{equation}
\label{eq:appF_2}
f(y|x) = f(y|0) + x f'_x(y|x_0), \quad y\in\set{Y}_1, \, |x|<\eps
\end{equation}
for some $0\leq x_0\leq x$, where we use the Lagrange form of the remainder. Consequently,
\begin{equation}
\label{eq:appF_3}
\log\frac{f(y|x)}{f(y|0)} = \log\left(1+x\frac{f'_x(y|x_0)}{f(y|0)}\right).
\end{equation}
By \eqref{eq:appD_1}, \eqref{eq:thm2_1}, and \eqref{eq:appD_bla}, it follows that
\begin{equation}
f(y|0) \geq \lambda_{\Delta} \quad \text{and} \quad |f'_x(y|x_0)| \leq \frac{1}{\Delta}\frac{1}{\sqrt{2\pi\sigma^2}} \label{eq:appF_4}
\end{equation}
for some $\lambda_{\Delta}>0$, which implies that
\begin{equation}
\left|\frac{f'_x(y|x_0)}{f(y|0)}\right| \leq \frac{1}{\Delta\lambda_{\Delta}}\frac{1}{\sqrt{2\pi\sigma^2}} \triangleq \kappa_{\Delta}, \quad y\in\set{Y}_1. \label{eq:appF_5}
\end{equation}
Using the inequality
\begin{equation}
|\log(1+x)| \leq \frac{|x|}{1-|x|}, \quad |x|<1 \label{eq:appF_6}
\end{equation}
it follows from \eqref{eq:appF_3}, \eqref{eq:appF_4}, and the monotonicity of $x\mapsto x/(1-x)$f that, for $|x|<\eps<1$,
\begin{IEEEeqnarray}{lCl}
\left|\log\frac{f(y|x)}{f(y|0)}\right| & = & \left|\log\left(1+x\frac{f'_x(y|x_0)}{f(y|0)}\right)\right| \nonumber\\
& \leq & \frac{|x|\left|\frac{f'_x(y|x_0)}{f(y|0)}\right|}{1-|x|\left|\frac{f'_x(y|x_0)}{f(y|0)}\right|} \nonumber\\
& \leq & \frac{\eps\kappa_{\Delta}}{1-\eps\kappa_{\Delta}}. \label{eq:appF_7}
\end{IEEEeqnarray}
The RHS of \eqref{eq:appF_7} vanishes as $\eps\downarrow 0$, so for any fixed $\delta>0$ and $\vartheta>0$, there exists an $\eps_0$ such that
\begin{equation}
\sup_{|x|<\eps} \left|\log\frac{f(y|x)}{f(y|0)}\right| < \delta, \quad \text{for all $\eps<\eps_0$.}
\end{equation}
Together with the definition of $\set{B}_{\epsilon,\delta}$ in \eqref{eq:Beps}, this proves \eqref{eq:appF_0}.

We continue by evaluating the integral for $y\in\set{Y}_2$. To this end, we use that the integrand is nonnegative and that $\set{B}_{\epsilon,\delta}\cap\set{Y}_2\subseteq \set{Y}_2$ to upper-bound
\begin{IEEEeqnarray}{lCl}
\int_{\set{B}_{\epsilon,\delta}\cap\set{Y}_2} \frac{[f'_x(y|x)]^2}{f(y|x)} \d y & \leq & \int_{\set{Y}_2} \frac{[f'_x(y|x)]^2}{f(y|x)} \d y. \label{eq:appF_8}
\end{IEEEeqnarray}
Repeating the steps \eqref{eq:appD_2}--\eqref{eq:appD_Y2} in Section~\ref{sub:Cond_E}, the RHS of \eqref{eq:appF_8} can be further upper-bounded by
\begin{equation}
\int_{\set{Y}_2} \frac{[f'_x(y|x)]^2}{f(y|x)} \d y \leq \sqrt{\frac{2}{\pi\sigma^2}} \frac{1}{\mu_{\Delta}(\vartheta)} e^{-\frac{(\vartheta-\eps-\Delta/2)^2}{2\sigma^2}} \label{eq:appF_Y2}
\end{equation}
where $\vartheta\mapsto\mu_{\Delta}(\vartheta)$ is a positive function that tends to $1$ as $\vartheta\to\infty$.

Combining \eqref{eq:appF_Y2} with \eqref{eq:appF_Y1} and \eqref{eq:appF_8}, we obtain that for every $\delta>0$ and sufficiently large $\vartheta>0$, there exists an $\eps_0$ such that
\begin{equation}
\int_{\set{B}_{\epsilon,\delta}} \frac{[f'_x(y|x))]^2}{f(y|x)} \d y \leq \sqrt{\frac{2}{\pi\sigma^2}} \frac{1}{\mu_{\Delta}(\vartheta)} e^{-\frac{(\vartheta-\eps-\Delta/2)^2}{2\sigma^2}}, \quad \text{for all $\eps<\eps_0$.}
\end{equation}
Upon integrating over $-\eps<x<\eps$, dividing by $\eps$, and taking the limit as $\eps\downarrow 0$, this yields for every $\delta>0$ and sufficiently large $\vartheta>0$
\begin{equation}
\lim_{\eps\downarrow 0} \frac{1}{\eps} \int_{-\eps}^{\eps} \int_{\set{B}_{\epsilon,\delta}} \frac{[f'_x(y|x))]^2}{f(y|x)} \d y\d x \leq \sqrt{\frac{2}{\pi\sigma^2}} \frac{2}{\mu_{\Delta}(\vartheta)} e^{-\frac{(\vartheta-\Delta/2)^2}{2\sigma^2}}.
\end{equation}
Condition~F follows then by letting $\vartheta$ tend to infinity.

\section*{Acknowledgment}
Stimulating discussions with Ram Zamir are gratefully acknowledged.



\end{document}